\documentclass[reqno, 10pt]{amsart}

\usepackage{amsmath, amsfonts, a4wide}
\usepackage[colorlinks=true]{hyperref}

\title[The ADHM construction and non-local symmetries]
{The ADHM construction and non-local symmetries of the self-dual Yang--Mills equations}
\author{James~D.E.\ Grant}
\date{22 December, 2009}

\numberwithin{equation}{section}

\theoremstyle{plain}
\newtheorem{theorem}{Theorem}[section]
\newtheorem{corollary}[theorem]{Corollary}
\newtheorem{proposition}[theorem]{Proposition}
\newtheorem{lemma}[theorem]{Lemma}

\theoremstyle{definition}

\newtheorem*{notation}{Notation}

\theoremstyle{remark}
\newtheorem{remark}[theorem]{Remark}

\newcommand{\sdyme}{self-dual Yang--Mills equations}

\newcommand{\cpone}{\ensuremath{\mathbb{C}P^1}}
\newcommand{\cptwo}{\ensuremath{\mathbb{C}P^2}}
\newcommand{\cpthree}{\ensuremath{{\mathbb{C}}P^3}}
\newcommand{\hpone}{\ensuremath{\mathbb{H}P^1}}

\newcommand{\Vzeroeps}{\ensuremath{\mathcal{V}^0_{\epsilon}}}
\newcommand{\Vinfeps}{\ensuremath{\mathcal{V}^{\infty}_{\epsilon}}}
\newcommand{\Veps}{\ensuremath{\mathcal{V}_{\epsilon}}}

\newcommand{\SUtwo}{\ensuremath{\mathrm{SU}_2}}
\newcommand{\sutwo}{\ensuremath{\mathfrak{su}_2}}
\newcommand{\SLtwoC}{\ensuremath{\mathrm{SL}_2(\mathbb{C})}}
\newcommand{\sltwoC}{\ensuremath{\mathfrak{sl}_2(\mathbb{C})}}

\newcommand{\Id}{\ensuremath{\mathrm{Id}}}

\newcommand\ubar{\ensuremath{\overline{u}}}
\newcommand\vbar{\ensuremath{\overline{v}}}

\newcommand\vs{\vskip .2cm}
\newcommand\bracket[2]{\ensuremath{\left\{ #1 \, \left| \vphantom{|^|} \, #2 \right.\right\}}}

\newcommand\commentout[1]{}

\begin{document}
\dedicatory{To Nicola Ramsay}
\thanks{This work was supported by START-project Y237--N13 of the Austrian Science Fund and a Visiting Professorship at the University of Vienna. The author is grateful to the anonymous referee, whose detailed comments led to the significant improvement of this paper.}
\address{\href{http://www.mat.univie.ac.at/home.php}{Fakult{\"a}t f{\"u}r Mathematik} \\
\href{http://www.univie.ac.at/de/}{Universit{\"a}t Wien} \\ Nordbergstrasse 15 \\ 1090 Wien \\ Austria}
\email{\href{mailto:james.grant@univie.ac.at}{james.grant@univie.ac.at}}

\begin{abstract}
We consider the action on instanton moduli spaces of the non-local symmetries of the self-dual Yang--Mills equations on $\mathbb{R}^4$ discovered by Chau and coauthors. Beginning with the ADHM construction, we show that a sub-algebra of the symmetry algebra generates the tangent space to the instanton moduli space at each point. We explicitly find the subgroup of the symmetry group that preserves the one-instanton moduli space. This action simply corresponds to a scaling of the moduli space.
\end{abstract}
\maketitle
\thispagestyle{empty}

\section{Introduction}
The {\sdyme} have been investigated from two rather distinct points of view in the last few decades. The first direction is in the study of the topology of four-manifolds, and the work of Donaldson (see, e.g.,~\cite{DK,FU}). In this approach, a fundamental role is played by the analysis of moduli spaces of solutions of the {\sdyme} with $L^2$ curvature (so-called \lq\lq instanton solutions\rq\rq) on a given four-manifolds. The analysis of such moduli-spaces then yields powerful information concerning differentiable structures on the underlying four-manifold. The second, seemingly unrelated, development is in the theory of integrable systems. In particular, it has been shown that many known integrable systems of differential equations may be derived as symmetry reductions of the {\sdyme} (see, e.g.,~\cite{MW}).

The purpose of the current paper, and its companion~\cite{reducible} which studies reducible connections, is to investigate whether properties of the {\sdyme} that follow from its integrable nature may be used to yield global information about instanton moduli spaces. In particular, it is known that the {\sdyme} on $\mathbb{R}^4$ admit an infinite-dimensional algebra of non-local symmetries~\cite{CGSW,CGW,CW}. In this paper, we investigate the action of these symmetries on the instanton-moduli spaces on $\mathbb{R}^4$ and, in particular, investigate, on the one-instanton moduli space, the sub-algebra of symmetries that preserve the $L^2$ condition on the curvature of the connection. Thinking of such symmetries as generating flows on the moduli space of all self-dual connections, $\mathcal{M}$, and of the $k$-instanton moduli space, $\mathcal{M}_k$, as a subspace of $\mathcal{M}$, then it is known that the non-local symmetries in general do not lie tangent to the subspaces $\mathcal{M}_k$ and, therefore, do not preserve the $L^2$ nature of the curvature (see, e.g.,~\cite[Chapter~V]{Cr}, but also Remark~\ref{counterexample} below). Our results are rather double-edged. In Theorem~\ref{TMk}, we show that the full tangent space to the moduli spaces $\mathcal{M}_k$ is generated by the fundamental vector fields of the symmetry algebra acting on the moduli space of self-dual connections $\mathcal{M}$. When we attempt to \lq\lq exponentiate\rq\rq\ these tangent vectors into a group action on $\mathcal{M}_k$, however, our conclusion is that the family of transformations that preserves the $L^2$ nature of the connection is rather small. In particular, the symmetries have orbits of rather high codimension in the moduli spaces. More specifically, in Theorem~\ref{scaling}, we deduce that the only symmetries of the {\sdyme} that act on five-dimensional one-instanton moduli space $\mathcal{M}_1$ correspond to a scaling of the instanton solutions. Such a collapse to orbits of large codimension is not unfamiliar from the theory of harmonic maps into Lie groups~\cite{AJS,GO,JS,UhlenbeckJDG}, where one has similar non-local symmetry algebras~\cite{D}.

\vs
The paper is organised as follows. In Section~\ref{sec:prelim}, we summarise the relevant background material that we require from both the integrable systems approach to the {\sdyme} and the ADHM approach to the instanton problem. Since our considerations are aimed at making a connection between the local aspects of the {\sdyme} and the global aspects, and the literature in these fields generally have completely different notation, it was deemed necessary to give an integrated, relatively detailed description of both approaches in a consistent notation. This accounts for the length of this section\footnote{For more information on the local aspects of the {\sdyme} that are relevant to us, see~\cite[Chapters~II \&~III]{Cr}. For more information concerning the ADHM formalism see either~\cite{ADHM} or~\cite[Chapters~II-IV]{At1}.}. In Section~\ref{sec:patchingADHM}, we show how the ADHM construction may be used to yield explicit patching matrices for holomorphic bundles over subsets of {\cpthree}, to which we may apply the results of~\cite{Cr} concerning the action of the symmetry algebra of the {\sdyme}. In Section~\ref{sec:deform}, we show that one-parameter families of ADHM data yield transformations that fall into the category of transformations considered in~\cite{Cr}, with the important proviso that these transformations are significantly restricted by the assumption that they are generated by flows on the full moduli space of self-dual connections. In Section~\ref{sec:one}, we show that our constructions can be carried out explicitly on the one-instanton moduli space, and that the only symmetries (consistent with a particular technical assumption) that have a well-defined action on the one-instanton moduli space are scalings. Finally, in an Appendix, we give a direct derivation of the action of the non-local symmetries of the {\sdyme} on the twistorial patching matrix from the action on the self-dual connection.

\vs
Finally, note that the symmetries that we investigate can also be constructed, by the same methods, on hyper-complex manifolds. Since we wish to consider symmetries that generalise to manifolds other than $\mathbb{R}^4$, and there exist hyper-complex manifolds with no continuous (conformal) isometries, we will not consider symmetries (such as those discussed in~\cite{PP}) that follow from the existence of a non-trivial conformal group on our manifold.

\section{Preliminaries}
\label{sec:prelim}
\subsection{Quaternions and twistor spaces}
We will deal exclusively with the {\sdyme} on $\mathbb{R}^4$ and $S^4$, and make constant use of isomorphisms $\mathbb{R}^4 \cong \mathbb{C}^2 \cong \mathbb{H}$, which we first fix. As such, let $\mathbf{x} := (x^1, x^2, x^3, x^4) \in \mathbb{R}^4$. We may then view
\[
u := \left( x^1 + i x^2 \right), \qquad
v := \left( x^3 - i x^4 \right)
\]
as coordinates on $\mathbb{C}^2$ and defining an isomorphism $\mathbb{R}^4 \rightarrow \mathbb{C}^2; \mathbf{x} \mapsto (u, v)$. In terms of these coordinates, the flat metric on $\mathbb{R}^4$ takes the form
\[
\mathbf{g} = \frac{1}{2} \left( du \otimes d{\ubar} + d{\ubar} \otimes du +
dv \otimes d{\vbar} + d{\vbar} \otimes dv \right).
\]
and the corresponding volume form is
\[
\epsilon = dx^1 \wedge dx^2 \wedge dx^3 \wedge dx^4 =
\frac{1}{4} du \wedge d{\ubar} \wedge dv \wedge d{\vbar}.
\]

Let $P \rightarrow \mathbb{R}^4$ be a principal {\SUtwo} bundle\footnote{We restrict to {\SUtwo} for simplicity. All of our considerations are equally valid for any classical Lie group.} over $\mathbb{R}^4$, and $E \rightarrow \mathbb{R}^4$ the rank-two complex vector bundle associated to $P$ via the fundamental representation of {\SUtwo}. (We will switch between the principal bundle and vector bundle pictures without comment.) An {\SUtwo} connection on $E$, $\mathbf{A} \in \Omega^1(\mathbb{R}^4, \sutwo)$ is a solution of the {\sdyme} if its curvature satisfies
\[
\star \mathbf{F} = \mathbf{F}.
\]
In terms of the complex coordinates introduced above, this equation is equivalent to
\begin{equation}
F_{uv} = 0, \qquad
F_{u {\ubar}} + F_{v {\vbar}} = 0, \qquad
F_{{\ubar}{\vbar}} = 0.
\label{sdym}
\end{equation}

We introduce complex vector fields $\mathbf{X}(z), \mathbf{Y}(z) \in C^{\infty}(\mathbb{R}^4, TM \otimes \mathbb{C})$ depending on a parameter $z \in \mathbb{C} \cup \{ \infty \} \equiv \cpone$:
\begin{equation}
\mathbf{X}(z) := \partial_{\ubar} - z \partial_v, \qquad
\mathbf{Y}(z) := \partial_{\vbar} + z \partial_u.
\label{XY}
\end{equation}
Then equations~\eqref{sdym} are equivalent to the requirement that
\begin{equation}
\mathbf{F} \left( \mathbf{X}(z), \mathbf{Y}(z) \right) = 0, \qquad \forall z \in \cpone.
\label{FXY}
\end{equation}

\vs

Since $\mathbb{R}^4 \subset \mathbb{R}^4 \cup \{ \infty \} \cong S^4 \cong \hpone$, a point $\mathbf{x} \in \mathbb{R}^4 \cong \mathbb{C}^2$ determines a quaternionic line in $\mathbb{H}^2$. In particular, we define $x := u + j v \in \mathbb{H}$. In terms of homogeneous coordinates $(p, q) \in \mathbb{H}^2$ on {\hpone}, the point $\mathbf{x}$ determines the quaternionic line
\[
l_{\mathbf{x}} := \bracket{(q, p) \in \mathbb{H}^2}{q = x p}
\]
in $\mathbb{H}^2$. Now, let $p = z_1 + j z_2$, $q = z_3 + j z_4$ with $\mathbf{z} := (z_1, \dots, z_4) \in \mathbb{C}^4$, and view $\mathbf{z}$ as homogeneous coordinates on {\cpthree}. Right-multiplication by $j$ on $\mathbb{H}^2$ defines an anti-linear anti-involution
\begin{equation}
\sigma\colon \mathbb{C}^4 \rightarrow \mathbb{C}^4; \qquad
\left( z_1, z_2, z_3, z_4 \right) \mapsto \left( -\overline{z_2}, \overline{z_1}, -\overline{z_4}, \overline{z_3} \right),
\label{C4involution}
\end{equation}
which descends to define an involution on the projective space:
\[
\sigma\colon \cpthree \rightarrow \cpthree; \qquad
\left[ z_1, z_2, z_3, z_4 \right] \mapsto \left[ -\overline{z_2}, \overline{z_1}, -\overline{z_4}, \overline{z_3} \right].
\]
The image of the quaternionic line $l_{\mathbf{x}}$ in {\cpthree} corresponding to $\mathbf{x} \in \mathbb{R}^4$ is given by the embedded projective line
\begin{equation}
L_{\mathbf{x}} \equiv L_{(u, v)} = \bracket{\left[ z_1, z_2, z_3, z_4 \right] \in \cpthree}{z_3 = z_1 u - z_2 {\vbar}, z_4 = z_1 v + z_2 {\ubar}}.
\label{Lx}
\end{equation}
Similarly, the embedded line corresponding to the point $\infty \in S^4$ is
\[
L_{\infty} :=
\bracket{\left[ 0, 0, z_3, z_4 \right]}{(z_3, z_4) \in \mathbb{C}^2 \setminus \left\{ (0, 0) \right\}}.
\]
The lines $L_p$, for $p \in S^4$, are invariant under the action of $\sigma$, and are referred to as \emph{real lines}. We will make particular use of the projection
\[
\pi\colon \cpthree \setminus L_{\infty} \rightarrow \mathbb{R}^4;
\qquad L_{\mathbf{x}} \to \mathbf{x}.
\]
On a fixed real line, $L_{\mathbf{x}}$, $\mathbf{x} \in \mathbb{R}^4$, we introduce the affine coordinate $z = z_2/z_1 \in \cpone$.

Finally, on the subset $\mathcal{U}_1 := \bracket{[z_1, z_2, z_3, z_4] \in \cpthree}{z_1 \neq 0}$, we may introduce coordinates $w_1 := z_3/z_1$, $w_2 := z_4/z_1$, $w_3 := z_2/z_1 \equiv z$. By definition, the coordinates $(w_1, w_2, w_3)$, viewed as functions on $\mathcal{U}_1$, are holomorphic with respect to the complex structure that $\mathcal{U}_1$ inherits as an open subset of {\cpthree}. From equation~\eqref{Lx}, the intersection $L_{\mathbf{x}} \cap \mathcal{U}_1$ consists of the set of points with $(w_1, w_2, w_3) = (u - z \vbar, v + z \ubar, z)$.  We will therefore often view $(u, v, z)$ as coordinates on the set $\mathcal{U}_1 \cong \mathbb{C}^2 \times \mathbb{C}$, with the functions $(u - z \vbar, v + z \ubar, z)$ being holomorphic with respect to the complex structure on $\mathcal{U}_1$. One may then check that, in this coordinate system, a basis for anti-holomorphic vector fields on the set $\mathcal{U}_1$ is given by the vector fields $\{ \mathbf{X}(z), \mathbf{Y}(z), \partial_{\overline{z}} \}$, with $\mathbf{X}(z), \mathbf{Y}(z)$ as in equation~\eqref{XY}. A similar argument may be carried out on the set $\mathcal{U}_2 := \bracket{[z_1, z_2, z_3, z_4] \in \cpthree}{z_2 \neq 0}$. In practice, the construction on $\mathcal{U}_2$ means that we may use the formulae for $\mathbf{X}(z), \mathbf{Y}(z)$ with $z$ taking values in {\cpone}. As such, we will often, henceforth, identify the set $\cpthree \setminus L_{\infty}$ with the set $\mathbb{C}^2 \times \cpone$. When we speak of a function being, for example, \lq\lq holomorphic\rq\rq\ on $\mathbb{C}^2 \times \mathbb{C} \subset \mathbb{C}^2 \times \cpone$, we will mean holomorphic on the set $\mathcal{U}_1$ with the induced complex structure mentioned above\footnote{From a {\cpone} point of view, we are viewing $\mathcal{U}_1 \cup \mathcal{U}_2 = \cpthree \setminus L_{\infty}$ as the total space of the normal bundle $\mathcal{O}(1) \oplus \mathcal{O}(1)$ of the rational curve $L_{\mathbf{0}} \subset \cpthree \setminus L_{\infty}$, where $\mathbf{0}$ denotes the origin in $\mathbb{R}^4$. $\mathbf{X}$ and $\mathbf{Y}$ are then linearly independent sections of this normal bundle. Unfortunately, this picture is not particularly well-suited to the calculations that we wish to perform.}.

\begin{notation}
Given $\epsilon > 0$, we define the following open subsets of {\cpone}:
\[
\Vzeroeps := \bracket{z \in \cpone}{|z| < 1 + \epsilon}, \qquad
\Vinfeps := \bracket{z \in \cpone}{|z| > \frac{1}{1 + \epsilon}},
\]
and their intersection
\[
\Veps := \Vzeroeps \cap \Vinfeps = \bracket{z \in \cpone}{\frac{1}{1 + \epsilon} < |z| < 1 + \epsilon}.
\]
We define the involution on the projective line $\sigma\colon \cpone \rightarrow \cpone; z \mapsto - \left. 1 \right/\!\overline{z}$ which, geometrically, is simply the anti-podal map. Note that $\sigma(\Vzeroeps) = \Vinfeps$ and $\sigma(\Veps) = \Veps$. Given any subset $\mathcal{V} \subset \cpone$ and a map $g\colon \mathcal{V} \rightarrow \SLtwoC$, we define a corresponding map $g^*\colon \sigma(\mathcal{V}) \rightarrow \SLtwoC$ by
\[
g^*(z) := \left( g(\sigma(z)) \right)^{\dagger}.
\]
(Throughout, $\dagger\colon \SLtwoC \rightarrow \SLtwoC$ will denote complex-conjugate transpose.) Similarly, given any map $f\colon U \times \mathcal{V} \rightarrow \SLtwoC$, we define a corresponding map $f^*\colon U \times \sigma(\mathcal{V}) \rightarrow \SLtwoC$ by
\[
f^*(x, z) := \left( f(x, \sigma(z)) \right)^{\dagger}.
\]
\end{notation}

\subsection{Holomorphic bundles}
\label{sec:patching}
An important property of the {\sdyme} is that they are the compatibility condition for the following overdetermined system of equations~\cite[Theorem~1]{Cr}
\begin{subequations}
\begin{align}
\left( \partial_{\ubar} - z \partial_v \right) \Psi(x, z) &=
- \left( A_{\ubar} - z A_v \right) \Psi(x, z),
\\
\left( \partial_{\vbar} + z \partial_u \right) \Psi(x, z) &=
- \left( A_{\vbar} + z A_u \right) \Psi(x, z),
\\
\partial_{\overline z} \Psi(x, z) &= 0,
\end{align}\label{ALP}\end{subequations}
for a map $\Psi\colon \mathbb{R}^4 \times \mathcal{V} \rightarrow \SLtwoC$, where $\mathcal{V}$ is a subset of {\cpone}. In particular, given $\epsilon > 0$, there exists a solution, $\Psi_0\colon \mathbb{R}^4 \times \Vzeroeps \rightarrow \SLtwoC$ that is analytic in $z$ for $z \in \Vzeroeps$. This solution is unique up to right multiplication
\[
\Psi_0(x, z) \rightarrow \widetilde{\Psi}_0(x, z) := \Psi_0(x, z) R(u - z\vbar, v + z \ubar, z),
\]
where $R\colon \mathbb{C}^2 \times \Vzeroeps \rightarrow \SLtwoC$ is holomorphic with respect to the complex structure that $\mathbb{C}^2 \times \Vzeroeps$ inherits as a subset of $\mathcal{U}_1$. Given $\Psi_0(x, z)$, we define
\[
\Psi_{\infty}(x, z) := \left( \Psi_0^*(x, z) \right)^{-1}.
\]
It is straightforward to check that $\Psi_{\infty}(x, z)$ is also a solution of~\eqref{ALP} that is analytic in $z$ on $\mathbb{R}^4 \times \Vinfeps$. Defining the fields
\[
\psi_0(x) := \Psi_0(x, 0), \qquad \psi_{\infty}(x) := \Psi_{\infty}(x, \infty),
\]
then equations~\eqref{ALP} imply that we may write the components of the connection in the form
\begin{subequations}
\begin{gather}
A_u = - \left( \partial_u \psi_{\infty}(x) \right) \psi_{\infty}(x)^{-1}, \qquad
A_v = - \left( \partial_v \psi_{\infty}(x) \right) \psi_{\infty}(x)^{-1}, \\
A_{\ubar} = - \left( \partial_{\ubar} \psi_0(x) \right) \psi_0(x)^{-1}, \qquad
A_{\vbar} = - \left( \partial_{\vbar} \psi_0(x) \right) \psi_0(x)^{-1}.
\end{gather}\label{A}\end{subequations}
We then define the Yang $J$-function $J\colon \mathbb{R}^4 \rightarrow \SLtwoC$ by
\begin{equation}
J := \psi_{\infty}(x)^{-1} \cdot \psi_0(x).
\label{Jdefn}
\end{equation}
Noting that, from the definition of $\Psi_{\infty}$, we have $\psi_{\infty}(x) = \left( \psi_0(x)^{\dagger} \right)^{-1}$, it then follows that $J^{\dagger} = J$. A short calculation shows that the remaining part of the anti-self-dual part of the curvature may be written in the form
\begin{align*}
F_{u\ubar} + F_{v\vbar}
&= - \psi_{\infty} \left[ \partial_u \left( J_{\ubar} J^{-1} \right) + \partial_v \left( J_{\vbar} J^{-1} \right) \right] \psi_{\infty}^{-1}
\\
&=  - \psi_0 \left[ \partial_{\ubar} \left( J^{-1} J_u \right) + \partial_{\vbar} \left( J^{-1} J_v \right) \right] \psi_0^{-1}.
\end{align*}
If the connection, $\mathbf{A}$, satisfies the {\sdyme} it therefore follows that the field $J$ satisfies the \emph{Yang--Pohlmeyer equation}
\begin{equation}
\partial_u \left( J_{\ubar} J^{-1} \right) +
\partial_v \left( J_{\vbar} J^{-1} \right) = 0.
\label{Jeqn}
\end{equation}
Conversely, given $J\colon \mathbb{R}^4 \rightarrow \SLtwoC$ that satisfies the Yang--Pohlmeyer equation and admits a splitting of the form~\eqref{Jdefn} for some $\psi_0$ and $\psi_{\infty}$ such that $\psi_{\infty} = \left( \psi_0^{-1} \right)^{\dagger}$, then the connection constructed as in Equations~\eqref{A} will satisfy the {\sdyme}.

Finally, the quantity
\begin{equation}
G(x, z) := \left( \Psi_{\infty}(x, z) \right)^{-1} \cdot \Psi_0(x, z),
\label{G}
\end{equation}
will be referred to as the \emph{patching matrix}. It defines a holomorphic map $\mathbb{C}^2 \times \Veps \subset \cpthree \setminus L_{\infty} \rightarrow \SLtwoC$, and hence the transition function of a holomorphic vector bundle over $\cpthree \setminus L_{\infty}$. The splitting~\eqref{G} implies that this bundle is trivial on restriction to real lines. The above is an explicit version of the Ward correspondence~\cite{Wa}, which defines a $1-1$ correspondence between self-dual Yang--Mills fields and holomorphic bundles over appropriate subsets of {\cpthree} that are trivial on restriction to real lines\footnote{See~\cite{Cr} for more details of the Ward correspondence from this point of view.}. Given such a bundle, the transition functions necessarily admit a splitting of the form~\eqref{G}, and the connection may then be reconstructed from the resulting fields $\Psi_0, \Psi_{\infty}$ via equations~\eqref{A}.

\subsection{Non-local symmetries of the {\sdyme}}
If we consider a one-parameter family of solutions, $J(t)$, of~\eqref{Jeqn}, depending in a $C^1$ fashion on a parameter $t \in (-\varepsilon, \varepsilon)$ then we deduce that $\dot{J} := \frac{d}{dt} J(t)$ must satisfy the linearisation of~\eqref{Jeqn}:
\begin{equation}
\partial_u \left( J \partial_{\ubar} \left( J^{-1} \dot{J} \right) J^{-1} \right) +
\partial_v \left( J \partial_{\vbar} \left( J^{-1} \dot{J} \right) J^{-1} \right) = 0.
\label{linearisation}
\end{equation}
Such a $\dot{J}$ defines a symmetry of the {\sdyme}. It is known that the only local symmetries\footnote{i.e. depending only on the connection and its derivatives pointwise} of the {\sdyme} on flat $\mathbb{R}^4$ are gauge transformations and those generated by the action of the conformal group (see, e.g., \cite{Po}). On the other hand, there exists a non-trivial family of non-local symmetries of the {\sdyme}~\cite{CGSW,CGW,CW}. To describe these, we define the auxiliary maps $\chi_0\colon \mathbb{R}^4 \times \Vzeroeps \rightarrow \SLtwoC$, $\chi_{\infty}\colon \mathbb{R}^4 \times \Vinfeps \rightarrow \SLtwoC$
\[
\chi_0(x, z) := \psi_0(x)^{-1} \cdot \Psi_0(x, z), \qquad
\chi_{\infty}(x, z) := \psi_{\infty}(x)^{-1} \cdot \Psi_{\infty}(x, z),
\]
which are analytic in $z$ for $z \in \Vzeroeps$ and $z \in \Vinfeps$, respectively. These maps have the property that $\chi_0(x, 0) = \chi_{\infty}(x, \infty) = \mathrm{Id}$.

\vs
The following result, based on the results of~\cite{CGSW,CGW,CW}, may then be extracted from Section~III.A of~\cite{Cr}:

\begin{proposition}
Let $T\colon \mathbb{R}^4 \times S^1 \rightarrow \SLtwoC$ be a map that extends continuously to a map $T\colon \mathbb{R}^4 \times \Veps \rightarrow \SLtwoC$ (for some $\epsilon > 0$) that is analytic in $z$ for $z \in \Veps$ and satisfies
\[
\left( \partial_{\vbar} + z \partial_u \right) T(x, z) =
\left( \partial_{\ubar} - z \partial_v \right) T(x, z) = 0
\]
for $(x, z) \in \mathbb{R}^4 \times \Veps$. Then, given any $\lambda \in \Veps$, the quantity
\begin{align}
\dot{J} :=& \chi_{\infty}(x, \lambda) T(x, \lambda) \chi_{\infty}(x, \lambda)^{-1} \cdot J + J \cdot \chi_0(x, \sigma(\lambda)) T(x, \lambda)^{\dagger} \chi_0(x, \sigma(\lambda))^{-1}
\nonumber
\\
=& \psi_{\infty}(x)^{-1}
\left[ \Psi_{\infty}(x, \lambda) T(x,  \lambda) \Psi_{\infty}(x,  \lambda)^{-1} + \right.
\nonumber
\\
&\hskip 3.5cm \left. \Psi_0(x, \sigma(\lambda)) T(x, \lambda)^{\dagger} \Psi_0(x, \sigma(\lambda))^{-1} \right]
\psi_0(x)
\label{Jsymmetry}
\end{align}
is a solution of the linearisation~\eqref{linearisation}.
\end{proposition}

In the case where the function $T$ is independent of $x$, it defines an element of the loop group $\Lambda \SLtwoC$ that admits a holomorphic extension to an open neighbourhood of $S^1$ in $\mathbb{C}^*$. The algebra of symmetries generated by such $T$ is then isomorphic to the Kac-Moody algebra of {\sltwoC}.

\vs
The action of such symmetries on the patching matrix is given by the following result:
\begin{theorem}
\label{thm:2.2}
Let $T\colon \mathbb{R}^4 \times S^1 \rightarrow \SLtwoC$ be as in the previous Proposition. The induced flow on the patching matrix of the corresponding bundle over $\cpthree \setminus L_{\infty}$ is given by
\begin{equation}
\dot{G}(x, z) = - T(x, z) G(x, z) - G(x, z) T^*(x, z) + \rho_{\infty}(x, z) G(x, z) + G(x, z) \rho_0(x, z),
\label{Gdot}
\end{equation}
for $(x, z) \in \mathbb{R}^4 \times \Veps$. In this equation, $\rho_0\colon \mathbb{R}^4 \times \Vzeroeps \rightarrow \sltwoC$ and $\rho_{\infty}\colon \mathbb{R}^4 \times \Vinfeps \rightarrow \sltwoC$ are analytic functions of $z$ on the respective regions and satisfy
\begin{align*}
\left( \partial_{\vbar} + z \partial_u \right) \rho_0(x, z) =
\left( \partial_{\ubar} - z \partial_v \right) \rho_0(x, z) = 0,
\\
\left( \partial_{\vbar} + z \partial_u \right) \rho_{\infty}(x, z) =
\left( \partial_{\ubar} - z \partial_v \right) \rho_{\infty}(x, z) = 0.
\end{align*}
Moreover, the functions $\rho_0$, $\rho_{\infty}$ may be absorbed into holomorphic changes of bases on the regions $\Vzeroeps$ and $\Vinfeps$. When this absortion process is carried out, the transformation~\eqref{Gdot} takes the simpler form
\begin{equation}
\dot{G}(x, z) = - T(x, z) G(x, z) - G(x, z) T^*(x, z).
\label{Gdot3}
\end{equation}
\end{theorem}

\begin{remark}
These transformations have been investigated from the viewpoint of twistor-theory and have a natural sheaf-theoretic interpretation~\cite{Iv2,Iv1,Pa,Po}. In the literature, it is standard to assume~\eqref{Gdot} (and the group-theoretic version~\eqref{crane} below) as the transformation law of the patching matrix, and to work backwards to derive the transformation law of $J$ and the connection (see, e.g.,~\cite{Iv2,Iv1,Po}). Since a direct proof of~\eqref{Gdot}, starting from~\eqref{Jsymmetry}, does not appear in the literature, we have included a proof in Appendix~\ref{App:symm}.
\end{remark}

\begin{remark}
The transformation~\eqref{Gdot3} is \emph{independent of the parameter $\lambda$} that appears in equation~\eqref{Jsymmetry}. As such, the transformation depends only on the function $T$. In equation~\eqref{Gdot}, the functions $\rho_0$, $\rho_{\infty}$ will generally depend on the parameter $\lambda$, but the corresponding dependence of $\dot{G}$ on $\lambda$ may be removed by a holomorphic change of frame. This situation is different from that in, for example, the theory of harmonic maps from a domain $X \subseteq \mathbb{R}^2$ to a Lie group $G$. In this case, one has a similar family of non-local symmetries~\cite{D} depending on a function $T(\lambda)$. There, however, the transformation properties of the extended harmonic map depends explicitly on the value of the parameter $\lambda$ (see, e.g.,~\cite[\S3]{UhlenbeckJDG}). Power-series expanding in $\lambda$ then gives a family of flows acting on the extended solution, and hence on the space of harmonic maps.
\end{remark}

The exponentiated form of the transformation law~\eqref{Gdot3} is given by the following:
\begin{theorem}\cite[Chapter~IV.C]{Cr}
\label{thm:Crane}
Let $g\colon \mathbb{R}^4 \times S^1 \rightarrow \SLtwoC$ be a smooth map that admits a continuous extension to a holomorphic map $g\colon \mathbb{C}^2 \times \Veps \subset \cpthree \setminus L_{\infty} \rightarrow \SLtwoC$, for some $\epsilon > 0$. Then we define the action of $g$ on the patching matrix $G(x, z)$ by the law
\begin{equation}
G(x, z) \mapsto g(x, z) \cdot G(x, z) \cdot g^*(x, z).
\label{crane}
\end{equation}
If $g$ extends holomorphically to $\mathbb{R}^4 \times \Vzeroeps$, then the corresponding transformation is a holomorphic change of basis on the bundle over $\cpthree \setminus L_{\infty}$, which leaves the self-dual connection, $\mathbf{A}$, unchanged.
\end{theorem}

\begin{remark}
The infinitesimal form of~\eqref{crane}, where $g(x, z) = \exp \left( - t T(x, z) \right)$ is equation~\eqref{Gdot3}.
\end{remark}

\subsection{The ADHM construction}
We wish to study the action of the symmetries mentioned above on the moduli spaces of instanton solutions of the {\sdyme} on $\mathbb{R}^4$ or, equivalently, $S^4$. As such, we are concerned with connections whose curvatures are $L^2$, in which case we have
\[
\int_{\mathbb{R}^4} | \mathbf{F} |^2 \, d^4 x = - 8 \pi^2 k,
\]
where $k \in \mathbb{N}_0$ is the second Chern number, $c_2(E)$, of the bundle $E$ (also called the instanton number of the connection). A self-dual connection with $L^2$ curvature on $\mathbb{R}^4$ necessarily extends to a self-dual connection on $S^4$~\cite{U}. We will refer to such connections, defined on either $\mathbb{R}^4$ or $S^4$ as an \emph{instanton}. The moduli space of instanton solutions, with instanton number $k$, modulo gauge transformations is a manifold of dimension $(8 k - 3)$ (away from singularities due to reducible connections), which we denote by $\mathcal{M}_k$. For later considerations, it will be important to think of $\mathcal{M}_k$ as being finite-dimensional submanifolds of the (infinite-dimensional) space of all self-dual connections on $\mathbb{R}^4$, not necessarily having $L^2$ curvature, which we denote by $\mathcal{M}$. The symmetries of the {\sdyme} that we consider may be viewed as defining flows on the space $\mathcal{M}$, and our main question is when these flows preserve the sub-manifolds $\mathcal{M}_k$.

\vs
Via the Ward correspondence~\cite{AHS,Wa}, self-dual connections of instanton number $k$ correspond to holomorphic bundles over {\cpthree} that are trivial on real lines. All such bundles may be constructed directly in terms of quaternionic linear algebra by the ADHM construction~\cite{ADHM}, which we now briefly recall.

For each $\mathbf{z} = (z_1, z_2, z_3, z_4) \in \mathbb{C}^4$, we define a linear map
\[
A(\mathbf{z})\colon W \rightarrow V,
\]
between complex vector spaces $W, V$ of dimension $k, 2k+2$ respectively, which is of the form
\[
A(\mathbf{z}) = \sum_{i = 1}^4 z_i A_i.
\]
The space $W$ is assumed to admit an anti-linear involution $\sigma_W\colon W \rightarrow W$. The space $V$ is assumed to have a symplectic form $\left( \cdot, \cdot \right)$ and an anti-linear anti-involution $\sigma_V\colon V \rightarrow V$ that are compatible in the sense that
\[
\left( \sigma_V u, \sigma_V v \right) = {\overline{\left( u, v \right)}}, \qquad
\forall \, u, v \in V.
\]
We require that the map $A(\mathbf{z})$ satisfies the compatibility condition
\begin{equation}
\sigma_V \!\left( A(\mathbf{z}) w \right) = A(\sigma(\mathbf{z}))\, \sigma_W(w), \qquad \forall \, w \in W, \qquad \forall \mathbf{z} \in \mathbb{C}^4,
\label{Areal}
\end{equation}
where $\sigma\colon \mathbb{C}^4 \rightarrow \mathbb{C}^4$ is as in equation~\eqref{C4involution}. Finally, we impose the following conditions:
\begin{itemize}
\item For all $\mathbf{z} \in \mathbb{C}^4$, the space $U_{\mathbf{z}} := A(\mathbf{z})(W) \subset V$ is of dimension $k$;
\item For all $\mathbf{z} \in \mathbb{C}^4$, $U_{\mathbf{z}}$ is isotropic with respect to $\left( \cdot, \cdot \right)$ i.e. $U_{\mathbf{z}} \subseteq U_{\mathbf{z}}^{\perp}$, where ${}^{\perp}$ denotes the complement with respect to the form $\left( \cdot, \cdot \right)$.
\end{itemize}

If we then define the quotient $E_{\mathbf{z}} := U_{\mathbf{z}}^{\perp} / U_{\mathbf{z}}$, then the collection of $E_{\mathbf{z}}$ defines a holomorphic, rank-$2$ complex vector bundle $E \rightarrow \cpthree$ with structure group $\SLtwoC$. The reality condition~\eqref{Areal} then imply that the bundle is trivial on restriction to any real line and that the self-dual connection on $S^4$ determined by the Ward correspondence is an {\SUtwo} connection.

\section{Patching matrix description of ADHM construction}
\label{sec:patchingADHM}

In order to make contact between the action of the non-local symmetries of the {\sdyme} in the form of~\eqref{Gdot} and the ADHM construction, we first need to reformulate the ADHM construction in terms of patching matrices.

We assume given an instanton solution of the {\sdyme} on $S^4$, with corresponding holomorphic bundle $E \rightarrow \cpthree$. We then consider (without any loss of information~\cite{U}) the restriction of this solution to $\mathbb{R}^4 \subset S^4$ and the restriction of the bundle $E$ to $\pi^{-1}(\mathbb{R}^4) \equiv \cpthree \setminus L_{\infty}$, which, for convenience, we denote by $E \rightarrow \cpthree \setminus L_{\infty}$. We split the set $\cpthree \setminus L_{\infty}$ as the union of two regions
\begin{align*}
\mathcal{S}_0 &:= \bracket{\left( (u, v), z \right) \in \mathbb{C}^2 \times \cpone}{|z| < 1 + \epsilon}
= \mathbb{C}^2 \times \Vzeroeps,
\\
\mathcal{S}_{\infty} &:= \bracket{\left( (u, v), z \right) \in \mathbb{C}^2 \times \cpone}{|z| > \frac{1}{1 + \epsilon}}
= \mathbb{C}^2 \times \Vinfeps.
\end{align*}
Since $\mathcal{S}_0, \mathcal{S}_{\infty} \cong \mathbb{C}^3$, the bundle $E$ restricted to either of these regions is holomorphically trivial~\cite{Cr,OSS}. The bundle $E$ is therefore characterised by the transition function $G\colon \mathcal{S}_0 \cap \mathcal{S}_{\infty} \rightarrow \SLtwoC$, which is the patching matrix from Section~\ref{sec:patching}.

The map $G$ may be constructed directly from the ADHM data, at the expense of fixing bases on the spaces $V$ and $W$. In particular, let $\{ \mathbf{a}_i \}_{i=1}^k$ be a basis of vectors in $W$ that are \emph{real} with respect to $\sigma_W$, in the sense that
\[
\sigma_W (\mathbf{a}_i) = \mathbf{a}_i, \qquad i = 1, \dots, k.
\]
(So, in practice, we are looking on $W$ as being the complexification of the fixed point set of $\sigma_W$.) The vectors
\[
\mathbf{v}_i(\mathbf{z}) := A(\mathbf{z}) \mathbf{a}_i \in V, \qquad i = 1, \dots, k
\]
define a collection of $k$ vectors that span the space $U_{\mathbf{z}} \subset V$. Due to the reality of the vectors $\mathbf{a}_i$, these vectors obey the reality condition
\begin{equation}
\sigma_V \!\left( \mathbf{v}_i (\mathbf{z}) \right) = \mathbf{v}_i (\sigma(\mathbf{z})), \qquad i = 1, \dots, k, \qquad \forall \mathbf{z} \in \mathbb{C}^4.
\label{vreality}
\end{equation}
Since $U_{\mathbf{z}}$ is isotropic with respect to the symplectic form, we deduce that
\begin{equation}
\left( \mathbf{v}_i (\mathbf{z}), \mathbf{v}_j (\mathbf{z}) \right) = 0, \qquad i, j = 1, \dots, k, \qquad \mathbf{z} \in \mathbb{C}^4.
\label{prodvv}
\end{equation}

We now view $\mathbf{z}$ as homogeneous coordinates on {\cpthree}, and construct bases for $\left[ \mathbf{z} \right] \in \mathcal{S}_0 \subset \cpthree \setminus L_{\infty}$. Given $\left[ \mathbf{z} \right] \in \mathcal{S}_0$, the annihilator $U_{\mathbf{z}}^{\perp}$ is spanned by $\{ \mathbf{v}_i (\mathbf{z}) \}$ along with two vectors $\bracket{\mathbf{e}_A (\mathbf{z})}{A = 1, 2}$ that span $U_{\mathbf{z}}^{\perp} / U_{\mathbf{z}}$. We therefore have
\begin{equation}
\left( \mathbf{v}_i (\mathbf{z}), \mathbf{e}_A (\mathbf{z}) \right) = 0, \qquad i = 1, \dots, k, \qquad A = 1, 2, \qquad \left[ \mathbf{z} \right] \in \mathcal{S}_0.
\label{prodve}
\end{equation}
and, without loss of generality, may assume that
\begin{equation}
\left( \mathbf{e}_1 (\mathbf{z}), \mathbf{e}_2 (\mathbf{z}) \right) = - \left( \mathbf{e}_2 (\mathbf{z}), \mathbf{e}_1 (\mathbf{z}) \right) = 1.
\label{prodee}
\end{equation}
Although not strictly necessary, it will sometimes be useful to extend the vectors $\{ \mathbf{e}_A (\mathbf{z}), \mathbf{v}_i (\mathbf{z}) \}$ to a full basis for $V$ by adding a set of vectors $\{\mathbf{w}^i (\mathbf{z})|i = 1, \dots, k\}$ with the property that
\begin{equation}
\left( \mathbf{w}^i (\mathbf{z}), \mathbf{w}^j (\mathbf{z}) \right) = 0, \qquad
\left( \mathbf{v}_i (\mathbf{z}), \mathbf{w}^j (\mathbf{z}) \right) = \delta_i^j, \qquad
\left( \mathbf{w}^i (\mathbf{z}), \mathbf{e}_A (\mathbf{z}) \right) = 0.
\label{prodw}
\end{equation}

We may also define a basis $\bracket{\mathbf{f}_A (\mathbf{z})}{A = 1, 2}$ for $U_{\mathbf{z}}^{\perp} / U_{\mathbf{z}}$ for $\left[ \mathbf{z} \right] \in \mathcal{S}_{\infty}$ by the relations
\[
\mathbf{f}_1 (\mathbf{z}) := - \sigma_V \!\left( \mathbf{e}_2 \!\left( \sigma(\mathbf{z}) \right) \right), \qquad
\mathbf{f}_2 (\mathbf{z}) := \sigma_V \!\left( \mathbf{e}_1 \!\left( \sigma(\mathbf{z}) \right) \right).
\]
This basis automatically has the property that
\begin{equation}
\left( \mathbf{f}_1(\mathbf{z}), \mathbf{f}_2(\mathbf{z}) \right) = 1, \qquad
\left( \mathbf{v}_i (\mathbf{z}), \mathbf{f}_A (\mathbf{z}) \right) = 0.
\label{prodff}
\end{equation}

Given that $\{ \mathbf{e}_A (\mathbf{z}) \}$ and $\{ \mathbf{f}_A (\mathbf{z}) \}$ are both bases for $U_{\mathbf{z}}^{\perp} / U_{\mathbf{z}}$ for $\left[ \mathbf{z} \right] \in \mathcal{S}_0 \cap \mathcal{S}_{\infty}$, there exist functions $G_A{}^B(\mathbf{z})$, $\lambda_A{}^i(\mathbf{z})$ defined on this region with the property that
\begin{equation}
\mathbf{f}_A(\mathbf{z}) = G_A{}^B(\mathbf{z}) \,\mathbf{e}_B(\mathbf{z}) + \lambda_A{}^i(\mathbf{z}) \,\mathbf{v}_i(\mathbf{z}).
\label{patchingmatrix}
\end{equation}
(From now on, summation convention will be assumed over repeated indices.) The matrix $G(\mathbf{z})$, defined for $\left[ \mathbf{z} \right] \in \mathcal{S}_0 \cap \mathcal{S}_{\infty}$ is then the transition function of our bundle $E$.

\

Before deriving some properties of the patching matrix that we will require, we define the {\SLtwoC}-invariant tensor $\epsilon$ by $\epsilon^{AB} = - \epsilon^{BA}$ with $\epsilon^{12} = 1$ and the $\mathrm{SO}_2(\mathbb{C})$-invariant tensor $\delta$ with components
\[
\delta_{AB} = \begin{cases} 1 &A = B, \\ 0 &A \neq B. \end{cases}
\]

\begin{proposition}
The patching matrix, $G$, as defined above obeys the conditions
\begin{gather*}
\det G(\mathbf{z}) = 1, \qquad
G^*(\mathbf{z}) = G(\mathbf{z}),
\end{gather*}
for $\left[ \mathbf{z} \right] \in \mathcal{S}_0 \cap \mathcal{S}_{\infty}$, where $G^*(\mathbf{z}) := G(\sigma(\mathbf{z}))^{\dagger}$. The functions $\lambda_A{}^i$ obey the reality condition
\[
{\overline{\lambda_A{}^i(\sigma(\mathbf{z}))}} = \delta_{AB} G_C{}^B \epsilon^{CD} \lambda_D{}^i(\mathbf{z}), \qquad
\lambda_A{}^i(\mathbf{z}) = - G_A{}^B \delta_{BC} \epsilon^{CD} {\overline{\lambda_D{}^i(\sigma(\mathbf{z}))}}
\]
for $\left[ \mathbf{z} \right] \in \mathcal{S}_0 \cap \mathcal{S}_{\infty}$.
\end{proposition}
\begin{proof}
Firstly, we have
\begin{align*}
1 &= \left( \mathbf{f}_1(\mathbf{z}), \mathbf{f}_2(\mathbf{z}) \right)
= \left( G_1{}^B(\mathbf{z}) \,\mathbf{e}_B(\mathbf{z}) + \lambda_1{}^i(\mathbf{z}) \,\mathbf{v}_i(\mathbf{z}), G_2{}^B(\mathbf{z}) \,\mathbf{e}_B(\mathbf{z}) + \lambda_2{}^i(\mathbf{z}) \,\mathbf{v}_i(\mathbf{z}) \right)
\\
&= \left( G_1{}^1(\mathbf{z}) G_2{}^2(\mathbf{z}) - G_1{}^2(\mathbf{z}) G_2{}^1(\mathbf{z}) \right) \left( \mathbf{e}_1(\mathbf{z}), \mathbf{e}_2(\mathbf{z}) \right)
= \det G(\mathbf{z}),
\end{align*}
where the four equalities follow from equations~\eqref{prodff}, \eqref{patchingmatrix}, \eqref{prodve} and~\eqref{prodee}, respectively. Therefore, $\det G(\mathbf{z}) = 1$, as required.

The definition of the vectors $\mathbf{f}_A(\mathbf{z})$ may be rewritten in the form
\begin{equation}
\mathbf{f}_A(\mathbf{z}) = - \delta_{AB} \epsilon^{BC} \sigma_V \!\left( \mathbf{e}_C \!\left( \sigma(\mathbf{z}) \right) \right).
\label{fdef}
\end{equation}
We now apply $\sigma_V$ to this equation, substitute equations~\eqref{patchingmatrix} and~\eqref{vreality}, and use the anti-linear, anti-involutive nature of $\sigma_V$. After some manipulation of $\delta$ and $\epsilon$ tensors, and using the fact that $\det G = 1$, we then find that
\[
\mathbf{f}_A(\mathbf{z}) = \left( G^*(\mathbf{z}) \right)_A{}^B \left[ \mathbf{e}_B(\mathbf{z}) - \delta_{BC} \epsilon^{CD} \overline{\lambda_D{}^i(\sigma(\mathbf{z}))} \mathbf{v}_i(\mathbf{z}) \right].
\]
Comparing with~\eqref{patchingmatrix} then gives the required equalities.
\end{proof}

\begin{remark}
We will be primarily interested in equation~\eqref{patchingmatrix} when it is restricted to a real line $L_{\mathbf{x}} \subset \cpone \setminus L_{\infty}$. Since the patching matrix, $G$, defined above is holomorphic on $\cpthree \setminus L_{\infty}$, when restricted to a neighbourhood of the line $L_{\mathbf{x}} \equiv L_{(u, v)}$, then $G$ will restrict to a function (which we denote by $G(\mathbf{x}, z)$) that is holomorphic in $(u-z\vbar, v+z\ubar, z)$ for $z \in \Veps$, for some $\epsilon > 0$.
\end{remark}

\section{One-parameter families of ADHM data}
\label{sec:deform}
We now consider a one-parameter family of ADHM data $A(\mathbf{z}) := A(t: \mathbf{z})$, with $t \in I$ a parameter, $I$ a sub-interval of the real line containing the origin. We assume that $A(t: \mathbf{z})$ is a $C^1$ function of $t$.

We wish to investigate how the elements of the above explicit construction depend on $A(t: \mathbf{z})$. The image $A(t: \mathbf{z}) \left( W \right)$ is now spanned by the vectors $\{ \mathbf{v}_i (t: \mathbf{z}) \}$, and $U_{\mathbf{z}}^{\perp} / U_{\mathbf{z}}$ is spanned by $\{ \mathbf{e}_A(t: \mathbf{z}) \}$, which we assume normalised such that~\eqref{prodee} is satisfied for each $t \in I$. Constructing the vectors $\{ \mathbf{f}_A(t: \mathbf{z}) \}$, we then define the patching matrix $G_A{}^B(t: \mathbf{z})$ and the functions $\lambda_A{}^i(t: \mathbf{z})$ as in~\eqref{patchingmatrix}.

\begin{proposition}
Given a one-parameter family of ADHM data, $A(t: \mathbf{z})$, and patching matrices as defined in~\eqref{patchingmatrix}, then there exists a matrix-valued function $d(t: z)$ with the property that
\begin{equation}
\dot{G}(t: \mathbf{z}) = d(t: \mathbf{z}) G(t: \mathbf{z}) + G(t: \mathbf{z}) d^*(t: \mathbf{z}).
\label{Gdot2}
\end{equation}
\end{proposition}
\begin{proof}
To investigate the $t$-dependence of these quantities, we consider their derivatives with respect to $t$. The derivatives of the relevant vectors are given as follows\footnote{Everything depends on $(t: \mathbf{z})$, but we drop explicit mention of this dependence for the moment}:
\begin{subequations}
\begin{align}
\dot{\mathbf{v}}_i &= A_i{}^j \mathbf{v}_j + B_{ij} \mathbf{w}^j + \epsilon^{AB} s_{Ai} \mathbf{e}_B,\\
\dot{\mathbf{w}}^i &= C^{ij} \mathbf{v}_j - A_j{}^i \mathbf{w}^j - \epsilon^{AB} r_A{}^i \mathbf{e}_B,\\
\dot{\mathbf{e}}_A &= c_A{}^B \mathbf{e}_B + r_A{}^i \mathbf{v}_i + s_{Ai} \mathbf{w}^i,
\end{align}\label{par1}\end{subequations}
where $A_i{}^j, \dots s_{Ai}$ are functions of $(t, \mathbf{z})$, that satisfy the relationships
\[
B_{ij} = B_{ji}, \quad
C^{ij} = C^{ji}, \quad
c_A{}^A = 0.
\]
It is straightforward to check that these are the most general forms of $\dot{\mathbf{v}}_i, \dot{\mathbf{w}}^i, \dot{\mathbf{e}}_A$ that preserve the relations~\eqref{prodvv}, \eqref{prodve}, \eqref{prodee} and~\eqref{prodw}.

We also define functions that characterise the time-dependence of the vector fields $\mathbf{f}_A$:
\begin{equation}
\dot{\mathbf{f}}_A = d_A{}^B \mathbf{f}_B + t_A{}^i \mathbf{v}_i + u_{Ai} \mathbf{w}^i.
\label{par2}
\end{equation}
{}From this expression and equation~\eqref{patchingmatrix}, we deduce that
\begin{equation}
\dot{G}_A{}^B = d_A{}^C G_C{}^B - G_A{}^C c_C{}^B + \lambda_A{}^i \epsilon^{BC} s_{Ci},
\label{Gdot1}
\end{equation}
along with the relations
\begin{align*}
\dot{\lambda}_A{}^i &= d_A{}^C \lambda_C{}^i + t_A{}^i - G_A{}^B r_B{}^i - \lambda_A^j A_j{}^i,\\
u_{Ai} &= G_A{}^B s_{Bi} + \lambda_A{}^j B_{ji}.
\end{align*}

Also, equating $\dot{\mathbf{f}}_1$ with $-\dot{{\overline{\mathbf{e}_2}}}$, and $\dot{\mathbf{f}}_2$ with $\dot{{\overline{\mathbf{e}_1}}}$, we find that
\[
d_A{}^B = u_{Ai} \lambda^{Bi} + \epsilon_{AC} \delta^{CD} {\overline{c_D{}^E}} \delta_{EF} \epsilon^{BF},
\]
and
\[
u_{Ai} = \epsilon_{AB} \delta^{BC} {\overline{s_{Ci}}}.
\]
These equations, along with~\eqref{Gdot1} imply that the $t$-derivative of the patching matrix obeys the relation~\eqref{Gdot2} with
\[
d = u_i \otimes \lambda^i + \epsilon_{AC} \delta^{CD} {\overline{c_D{}^E}} \delta_{EF} \epsilon^{BF},
\]
as required.
\end{proof}

\begin{remark}
The quantities that occur in equation~\eqref{Gdot2} may all be constructed directly from the vector fields $\mathbf{e}_A, \mathbf{v}_i$ since
\[
\left( \mathbf{v}_i, \dot{\mathbf{e}}_A \right) = s_{Ai}, \qquad
\left( \dot{\mathbf{e}}_A, \mathbf{e}_B \right) = \sum_C c_A{}^C \epsilon_{CB}.
\]
Therefore the construction does not actually require the introduction of the basis vectors $\{ \mathbf{w}^i \}$.
\end{remark}

\begin{corollary}
Given a one-parameter family of ADHM data and patching matrix defined as above, then there exists a map $d\colon I \times \mathbb{C}^2 \times \Veps \rightarrow \SLtwoC$ that is holomorphic in $(u-z\vbar, v+z\ubar, z)$ for $z \in \Veps$ such that the restriction of the patching matrix to real-lines $L_{\mathbf{x}}$ evolves according to
\begin{equation}
\dot{G}(t: \mathbf{x}, z) = d(t: \mathbf{x}, z) G(t: \mathbf{x}, z) + G(t: \mathbf{x}, z) d^*(t: \mathbf{x}, z),
\label{Gdotline}
\end{equation}
for $(\mathbf{x}, z) \in \mathbb{C}^2 \times \Veps$.
\end{corollary}
\begin{proof}
Restrict~\eqref{Gdot2} to $L_{\mathbf{x}}$.
\end{proof}

\begin{remark}
Let $\alpha(t: \mathbf{x}, z)$ satisfy the first order ordinary differential equation
\[
\dot{\alpha}(t: \mathbf{x}, z) = d(t: \mathbf{x}, z) \alpha(t: \mathbf{x}, z), \qquad
\alpha(0: \mathbf{x}, z) = \Id.
\]
Given an initial patching matrix $G(\mathbf{x}, z)$, it follows that the one-parameter family of patching matrices
\begin{equation}
G(t: \mathbf{x}, z) := \alpha(t: \mathbf{x}, z) \, G(\mathbf{x}, z) \, \alpha^*(t: \mathbf{x}, z)
\label{Gnew}
\end{equation}
satisfies equation~\eqref{Gdot2} with initial conditions $G(0: \mathbf{x}, z) = G(\mathbf{x}, z)$. Conversely, by uniqueness of solutions of~\eqref{Gdot2}, it follows that $G(t: \mathbf{x}, z)$, as defined in equation~\eqref{Gnew}, is the unique one-parameter family of patching matrices determined by the flow~\eqref{Gdot2} with initial data $G(\mathbf{x}, z)$.
\end{remark}

\vs
Note that these transformation~\eqref{Gdotline} and~\eqref{Gnew} are of the same form as those generated by the symmetries of the {\sdyme} given in equation~\eqref{Gdot} and Theorem~\ref{thm:Crane}, with the important proviso that the function $d(t: x, z)$ occurring in~\eqref{Gdotline} depends explicitly on the parameter $t$. The symmetries~\eqref{Gdot} should be viewed as defining a flow on the space, $\mathcal{M}$, of self-dual connections defined by the map $T$. In solving~\eqref{Gdot}, we are simply constructing the integral curves of this flow, with $t$ a parameter along the integral curve. As such, in~\eqref{Gdot}, it is important that the function $T(x, z)$ is independent of the parameter $t$.

Viewing the function $T$ as defining a flow on $\mathcal{M}$ and the instanton moduli spaces $\mathcal{M}_k$ as submanifolds of $\mathcal{M}$, we directly deduce:

\begin{theorem}
Let $\mathbf{A} \in \mathcal{M}_k$ be a $k$-instanton self-dual connection (modulo gauge transformation) on $\mathbb{R}^4$, with $\mathcal{M}_k$ viewed as a submanifold of the space, $\mathcal{M}$, of all self-dual connections on $\mathbb{R}^4$. Then for each vector $\mathbf{v} \in T_{\mathbf{A}} \mathcal{M}_k$, there exists a function $T$ such that the fundamental vector field on $\mathcal{M}$ corresponding to $T$ via equation~\eqref{Gdot} consider with $\mathbf{v}$ at the point $\mathbf{A} \in \mathcal{M}_k$.
\label{TMk}
\end{theorem}
\begin{proof}
Any element $\mathbf{v} \in T_{\mathbf{A}} \mathcal{M}_k$ is generated by a one-parameter family of ADHM data, $A(t: \mathbf{z})$, with $A(0: \mathbf{z})$ corresponding to the connection $\mathbf{A}$. This one-parameter family of ADHM data then gives rise to a one-parameter family of patching matrices $G(t: \mathbf{x}, z)$ evolving according to~\eqref{Gdotline}, where $G(0: \mathbf{x}, z)$ is the patching matrix corresponding to $\mathbf{A}$ and $\dot{G}(0: \mathbf{x}, z)$ corresponds to the tangent vector $\mathbf{v}$. Taking $T(x, z) := - d(0: x, z)$ gives a symmetry that, via~\eqref{Gdot} (with $\rho_0 = \rho_{\infty} = 0$) generates the tangent vector $\mathbf{v}$.
\end{proof}

\begin{remark}
\label{rem:absorbtion}
Theorem~\ref{thm:2.2} states that, given a function $T$, there is a corresponding fundamental vector field on $\mathcal{M}$, the space of self-dual connections, corresponding to $T$. We shall denote this fundamental vector field by $\mathbf{X}_T$. Theorem~\ref{TMk} states that, given a connection $\mathbf{A} \in \mathcal{M}_k$ and a tangent vector $\mathbf{v} \in T_{\mathbf{A}} \mathcal{M}_k$, then there exists such a function $T$ such that $\left. \mathbf{X}_T \right|_{\mathbf{A}} = \mathbf{v}$. It is important to note, however, that the integral curve of $\mathbf{X}_T$ starting at $\mathbf{A} \in \mathcal{M}_k$ will, generally, \emph{not\/} remain within the sub-manifold $\mathcal{M}_k$ of $\mathcal{M}$. In order to determine which one-parameter groups of symmetries gives flows that remain in the moduli space $\mathcal{M}_k$, we need to determine which transformations of the form~\eqref{Gdotline} are generated by transformations of the form~\eqref{Gdot}, with $T(x, z)$ independent of $t$.

From the form of~\eqref{Gdotline} and~\eqref{Gdot}, it appears natural to identify $d(t: x, z)$ with $-T(x, z) + \rho_{\infty}(t, x, z)$. We impose that $T$ is independent of $t$. The map $\rho_{\infty}$ simply generates a change of holomorphic frame for $z \in \Vinfeps$. At this point, we should recall that we have partially fixed our holomorphic frames in deriving our patching matrix from the ADHM data. As such, if we wish to employ our approach with one-parameter families of ADHM data, we must allow for one-parameter families of changes of frame in order to compensate for this fixing of frames. As such, we should allow $\rho_{\infty}$ to be $t$-dependent (i.e. $\rho_{\infty} = \rho_{\infty}(t, x, z)$). Note that such a $t$-dependent change of frame does not affect the corresponding self-dual connection $\mathbf{A}(t)$.

As such, we may use $\rho_{\infty}$ to absorb any part of $d(t: x, z)$ that is holomorphic on $\mathbb{C}^2 \times \Vinfeps$, leaving an irreducible part of $d(t: x, z)$, denoted $d_0(t: x, z)$, that has singularities in the region $\Vinfeps$ that cannot be removed by absorption into $\rho_{\infty}$. In order to arise from a symmetry of the {\sdyme}, $d_0(t: x, z)$ must then be independent of $t$. Since $d(t: x, z)$ is determined by first $t$-derivatives of the ADHM data, $A(t: \mathbf{z})$, imposing that $d_0(t: x, z)$ is constant in $t$ will impose conditions on the first $t$-derivatives of the $A(t: \mathbf{z})$ data that must be satisfied in order for this one-parameter family of data (and corresponding self-dual connections) to arise from a symmetry of the {\sdyme}. Explicit calculations, in the next section, suggest that these conditions are quite restrictive.
\end{remark}

\begin{remark}
The fact that the flow on the moduli space does not generally preserve the $L^2$ nature of the curvature is well-known (see, e.g., \cite{CGSW,CGW,CW} where this effect is mentioned). In~\cite[Chapter V]{Cr}, an explicit example of a transformation acting on a one-instanton patching matrix is given to demonstrate this phenomenon. In the notation of~\eqref{Gnew}, this transformation takes the form
\[
\alpha(t: x, z) = \frac{1}{\sqrt{1-t^2}} \begin{pmatrix} 1 &t/z \\ tz &1 \end{pmatrix}.
\]
From this expression, we deduce that
\[
d(t: x, z) = \frac{1}{\left( 1-t^2 \right)^{3/2}} \begin{pmatrix} t &1/z \\ z &t \end{pmatrix}.
\]
Following the programme of the previous remark, we then isolate the part of $d$ that has singularities in the region $z \in \Vinfeps$, namely
\[
d_0(t: x, z) = \frac{1}{\left( 1-t^2 \right)^{3/2}} \begin{pmatrix} 0 &0 \\ z &0 \end{pmatrix}.
\]
Since $d_0$ depends explicitly on $t$, we deduce that the counterexample provided in~\cite[Chapter V]{Cr} falls outside of the class of transformations generated by transformations~\eqref{Gdot} with $T$ independent of $t$.
\label{counterexample}
\end{remark}

\begin{remark}
If one drops the reality condition that our connections are {\SUtwo} connections, rather than {\SLtwoC} connections, then Takasaki has argued~\cite{T} that the action of the non-local symmetry group generated by transformations of the form
\[
\dot{J}(x) = \chi_{\infty}(x, \lambda) T(x, \lambda) \chi_{\infty}(x, \lambda)^{-1} \cdot J
\]
is transitive on the space of {\SLtwoC} solutions of the {\sdyme}. If, as here, we restrict to symmetries of the form~\eqref{Jsymmetry} that explicitly preserve the {\SUtwo} nature of the connection, then the symmetry group need not act transitively on the moduli space of solutions, even though the symmetries have been shown to generate the tangent space at each point. Moreover, if we explicitly impose that we only consider symmetries that preserve the $L^2$ nature of the connection, then the explicit calculations carried out in the next Section for the one-instanton moduli space suggest that the orbits of the symmetry group are actually of high codimension in the moduli space.
\end{remark}

\section{The one-instanton solution}
\label{sec:one}

In the case of a one-instanton solution, it is straightforward to carry out the ADHM construction and the construction of deformations explicitly. We find that the one-parameter subgroups of ADHM data with $d(t: x, z)$ of the form $-T(x, z) + \rho_{\infty}(t, x, z)$ are rather small.

First, we fix some notation. In the case $k=1$, then we may write
\[
\mathbf{v}(\mathbf{z}) := A(\mathbf{z}) = \begin{pmatrix} A_1(\mathbf{z}) \\ A_2(\mathbf{z}) \\ A_3(\mathbf{z}) \\ A_4(\mathbf{z}) \end{pmatrix} \in \mathbb{C}^4,
\]
where $A_i(\mathbf{z}) = \sum_{j=1}^4 A_i^j z_j, i = 1, \dots, 4$. Letting
\[
\sigma_V \begin{pmatrix}\alpha \\ \beta \\ \gamma \\ \delta \end{pmatrix} := \begin{pmatrix} - \overline{\beta} \\ \overline{\alpha} \\ - \overline{\delta} \\ \overline{\gamma} \end{pmatrix},
\]
then~\eqref{vreality} implies that the functions $A_i(\mathbf{z})$ must satisfy the reality conditions:
\begin{align*}
{\overline{A_1(\sigma(\mathbf{z}))}} &= - A_2(\mathbf{z}),
&{\overline{A_2(\sigma(\mathbf{z}))}} &= A_1(\mathbf{z}),
\\
{\overline{A_3(\sigma(\mathbf{z}))}} &= - A_4(\mathbf{z}),
&{\overline{A_4(\sigma(\mathbf{z}))}} &= A_3(\mathbf{z}).
\end{align*}
In particular, using the symmetry transformations inherent in the ADHM construction~\cite[Chapter~II]{At1}, we may fix
\begin{subequations}
\begin{align}
A_1(\mathbf{z}) &= \lambda z_1,
&A_2(\mathbf{z}) &= \lambda z_2,
\\
A_3(\mathbf{z}) &= \alpha z_1 - {\overline\beta} z_2 - z_3,
&A_4(\mathbf{z}) &= \beta z_1 + {\overline\alpha} z_2 - z_4,
\end{align}\label{lambdaetc}\end{subequations}
where $\lambda$ is a positive, real number and $\alpha, \beta$ are complex numbers. Finally, we may take the symplectic form on $V \cong \mathbb{C}^4$ to be
\[
\left( \mathbf{a}, \mathbf{b} \right) = a^1 b^2 - a^2 b^1 + a^3 b^4 - a^4 b^3, \qquad
\mathbf{a}, \mathbf{b} \in \mathbb{C}^4.
\]

\begin{theorem}
The only transformations of the ADHM data $(\lambda, \alpha, \beta)$ that arise from a non-local symmetry of the {\sdyme}~\eqref{Jsymmetry} according to~\eqref{Gdotline} with $d(t: x, z)$ of the form $-T(x, z) + \rho_{\infty}(t: x, z)$ are of the form
\begin{equation}
\lambda \mapsto \lambda(t) := \frac{\lambda}{\sqrt{1-k\lambda^2 t}}, \qquad \alpha, \beta \mbox{ constant},
\label{lambdatransformation}
\end{equation}
where $k \in \mathbb{R}$ is a real constant.
\label{scaling}
\end{theorem}

\begin{proof}
On a region with $A_1(\mathbf{z}) \neq 0$ (and hence $A_2(\mathbf{z}) \neq 0$), then we find that $U_{\mathbf{z}} = \mathbf{v}(\mathbf{z})^{\perp} / \mathbf{v}$ is spanned by the vectors
\[
\mathbf{e}_1(\mathbf{z}) = \left( 0, \frac{A_4(\mathbf{z})}{A_1(\mathbf{z})}, 1, 0 \right), \qquad
\mathbf{e}_2(\mathbf{z}) = \left( 0, - \frac{A_3(\mathbf{z})}{A_1(\mathbf{z})}, 0, 1 \right),
\]
which have the property that $\left( \mathbf{e}_1, \mathbf{e}_2 \right) = 1$. Such a basis, including the normalisation property, is unique up to a translation $\mathbf{e}_A \mapsto \mathbf{e}_A + \lambda_A \mathbf{v}$, and an {\SLtwoC} rotation of the vectors $\mathbf{e}_A(\mathbf{z})$. Taking the conjugates of these vectors, we find that
\[
\mathbf{f}_1(\mathbf{z}) = - {\overline{\mathbf{e}_2(\mathbf{z})}} = \left( - \frac{A_4(\mathbf{z})}{A_2(\mathbf{z})}, 0, 1, 0 \right), \qquad
\mathbf{f}_2(\mathbf{z}) = {\overline{\mathbf{e}_1(\mathbf{z})}} = \left( \frac{A_3(\mathbf{z})}{A_2(\mathbf{z})}, 0, 0, 1 \right).
\]
These expressions imply that on the overlap where the two above regions overlap, we have the patching matrix (see~\cite[Chapter~V]{Cr})
\[
G = \begin{pmatrix} 1 + \frac{A_3(\mathbf{z}) A_4(\mathbf{z})}{A_1(\mathbf{z}) A_2(\mathbf{z})} & \frac{A_4(\mathbf{z})^2}{A_1(\mathbf{z}) A_2(\mathbf{z})} \\
- \frac{A_3(\mathbf{z})^2}{A_1(\mathbf{z}) A_2(\mathbf{z})} &1 - \frac{A_3(\mathbf{z}) A_4(\mathbf{z})}{A_1(\mathbf{z}) A_2(\mathbf{z})} \end{pmatrix}
\]
and
\[
\lambda_1(\mathbf{z}) = - \frac{A_4(\mathbf{z})}{A_1(\mathbf{z}) A_2(\mathbf{z})}, \qquad
\lambda_2(\mathbf{z})= \frac{A_3(\mathbf{z})}{A_1(\mathbf{z}) A_2(\mathbf{z})}.
\]
We may take the vector $\mathbf{w}(\mathbf{z})$ to be
\[
\mathbf{w} = \left( 0, \frac{1}{A_1(\mathbf{z})}, 0, 0 \right),
\]
which is unique up to $\mathbf{w} \mapsto \mathbf{w} + \phi \mathbf{v}$.

If we now let $\mathbf{v}(\mathbf{z})$ depend smoothly on a parameter $t \in \left( - \epsilon, \epsilon \right)$, then we may calculate the parameters of the deformation $A, B, C, \dots$ as defined in~\eqref{par1} and~\eqref{par2}. The parameter $d$ is the one that we primarily require and a straightforward calculation shows that
\commentout{We find that the only non-zero parameters are:
\begin{gather*}
A(t: \mathbf{z}) = \frac{A_1^{\prime}(t, \mathbf{z})}{A_1(t, \mathbf{z})},\\
B(t: \mathbf{z}) = A_1(t, \mathbf{z}) A_2^{\prime}(t, \mathbf{z}) - A_2(t, \mathbf{z}) A_1^{\prime}(t, \mathbf{z}) + A_3(t, \mathbf{z}) A_4^{\prime}(t, \mathbf{z}) - A_4(t, \mathbf{z}) A_3^{\prime}(t, \mathbf{z}),\\
C = 0,\\
c_A{}^B = 0,\\
r_A = 0,\\
s_1(t, \mathbf{z}) = A_1(t, \mathbf{z}) \left( \frac{A_4(t, \mathbf{z})}{A_1(t, \mathbf{z})} \right)^{\prime}, \qquad
s_2(t, \mathbf{z}) = - A_1(t, \mathbf{z}) \left( \frac{A_3(t, \mathbf{z})}{A_1(t, \mathbf{z})} \right)^{\prime},\\
t_1(t, \mathbf{z}) = - \frac{1}{A_1(t, \mathbf{z})} \left( \frac{A_4(t, \mathbf{z})}{A_2(t, \mathbf{z})} \right)^{\prime}, \qquad
t_2(t, \mathbf{z}) = \frac{1}{A_1(t, \mathbf{z})} \left( \frac{A_3(t, \mathbf{z})}{A_2(t, \mathbf{z})} \right)^{\prime},\\
u_1(t, \mathbf{z}) = A_2(t, \mathbf{z}) \left( \frac{A_4(t, \mathbf{z})}{A_2(t, \mathbf{z})} \right)^{\prime}, \qquad
u_2(t, \mathbf{z}) = - A_2(t, \mathbf{z}) \left( \frac{A_3(t, \mathbf{z})}{A_2(t, \mathbf{z})} \right)^{\prime},
\end{gather*}
It follows from these formulae that}
\[
d(t: \mathbf{z}) =
\frac{\partial}{\partial t}\! \begin{pmatrix} A_4(t: \mathbf{z}) / A_2(t: \mathbf{z}) \\ - A_3(t: \mathbf{z}) / A_2(t: \mathbf{z}) \end{pmatrix}
\times \begin{pmatrix} A_3(t: \mathbf{z}) / A_1(t: \mathbf{z}) & A_4(t: \mathbf{z}) / A_1(t: \mathbf{z}) \end{pmatrix}
\]
Taking $A_i(t: \mathbf{z})$ as in~\eqref{lambdaetc}, with $\lambda$ replaced by $\lambda(t)$, etc, then, restricted to the line $L_{\mathbf{x}}$, the deformation parameter that we require takes the form
\[
d(\mathbf{x}, z) = \frac{1}{z}
\frac{\partial}{\partial t}\!\begin{pmatrix}
\frac{(\beta - v) + z ({\overline\alpha} - {\ubar})}{\lambda}
\\
- \frac{(\alpha -u) - z ({\overline\beta} - {\vbar})}{\lambda}
\end{pmatrix}
\times \begin{pmatrix} \frac{(\alpha -u) - z ({\overline\beta} - {\overline b})}{\lambda}
& \frac{(\beta - v) + z ({\overline\alpha} - {\ubar})}{\lambda}
\end{pmatrix}.
\]
This expression may be written in the form
\begin{align*}
d(\mathbf{x}, z) =& \frac{1}{z}
\left[
A (u - z {\vbar})^2 + B (u - z {\vbar}) (v + z {\ubar}) + C (v + z {\ubar})^2
\right]
\\
&
+ \left( \frac{D}{z} + E \right) (u - z {\vbar})
+ \left( \frac{F}{z} + G \right) (v + z {\ubar})
+ \frac{H}{z} + I + J z,
\end{align*}
where
\begin{gather*}
A = \frac{\dot{\lambda}}{\lambda^3} \begin{pmatrix} 0 & 0 \\ 1 & 0 \end{pmatrix}
\qquad
B = \frac{\dot{\lambda}}{\lambda^3} \begin{pmatrix} -1 & 0 \\ 0 & 1 \end{pmatrix}
\qquad
C = \frac{\dot{\lambda}}{\lambda^3} \begin{pmatrix} 0 & -1 \\ 0 & 0 \end{pmatrix}
\\
D = \frac{1}{\lambda^3}
\begin{pmatrix}
 - \lambda \dot{\beta} + \beta \dot{\lambda} & 0
\\
\lambda \dot{\alpha} - 2 \alpha \dot{\lambda}
& - \beta \dot{\lambda}
\end{pmatrix}
\qquad
E = \frac{1}{\lambda^3}
\begin{pmatrix}
 - \lambda \dot{\overline\alpha} + {\overline\alpha} \dot{\lambda} & 0
\\
 - \lambda \dot{\overline\beta} + 2 {\overline\beta} \dot{\lambda}
& - {\overline\alpha} \dot{\lambda}
\end{pmatrix}
\\
F = \frac{1}{\lambda^3}
\begin{pmatrix}
\alpha \dot{\lambda} & - \lambda \dot{\beta} + 2 \beta \dot{\lambda}
\\
0 & \lambda \dot{\alpha} - \alpha \dot{\lambda}
\end{pmatrix}
\qquad
G = \frac{1}{\lambda^3}
\begin{pmatrix}
 - {\overline\beta} \dot{\lambda} & - \lambda \dot{\overline\alpha} + 2 {\overline\alpha} \dot{\lambda}
\\
0 & - \lambda \dot{\overline\beta} + {\overline\beta} \dot{\lambda}
\end{pmatrix}
\\
H = \frac{1}{\lambda^3}
\begin{pmatrix}
\alpha (\lambda \dot{\beta} - \beta \dot{\lambda})
& \beta (\lambda \dot{\beta} - \beta \dot{\lambda})
\\
\alpha ( - \lambda \dot{\alpha} + \alpha \dot{\lambda})
& \beta ( - \lambda \dot{\alpha} + \alpha \dot{\lambda})
\end{pmatrix}
\\
I = \frac{1}{\lambda^3}
\begin{pmatrix}
\alpha \lambda \dot{\overline\alpha} - {\overline\beta} \lambda \dot{\beta} - \alpha {\overline\alpha} \dot{\lambda} + \beta
 {\overline\beta} \dot{\lambda}
&
\beta \lambda \dot{\overline\alpha} + {\overline\alpha} \lambda \dot{\beta} - 2 {\overline\alpha} \beta \dot{\lambda}
\\
{\overline\beta} \lambda \dot{\alpha} + \alpha \lambda \dot{\overline\beta} - 2 \alpha {\overline\beta} \dot{\lambda}
& {\overline\alpha} ( - \lambda \dot{\alpha} + \alpha \dot{\lambda}) + \beta (\lambda \dot{\overline\beta} - {\overline\beta} \dot{\lambda})
\end{pmatrix}
\\
J = \frac{1}{\lambda^3}
\begin{pmatrix}
{\overline\beta} ( - \lambda \dot{\overline\alpha} + {\overline\alpha} \dot{\lambda})
& {\overline\alpha} (\lambda \dot{\overline\alpha} - {\overline\alpha} \dot{\lambda})
\\
{\overline\beta} ( - \lambda \dot{\overline\beta} + {\overline\beta} \dot{\lambda})
& {\overline\alpha} (\lambda \dot{\overline\beta} - {\overline\beta} \dot{\lambda}),
\end{pmatrix}
\end{gather*}
where $\dot{}$ denotes differentiation with respect to $t$.
\vs
According to the philosophy of Remark~\ref{rem:absorbtion}, we note that the coefficients $D$, $F$, $H$ and $I$ correspond to terms that are analytic for $z \in \Vinfeps$, and therefore may be absorbed into the $\rho_{\infty}$ term. The remaining part of the parameter $d$ is then
\begin{align*}
d_0(\mathbf{x}, z) &= \frac{1}{z}
\left[ A (u - z {\vbar})^2 + B (u - z {\vbar}) (v + z {\ubar}) + C (v + z {\ubar})^2
\right]
\\
&\hskip 2cm + E (u - z {\vbar}) +  G (v + z {\ubar}) + J z.
\end{align*}
All of the terms in $d_0$ have singularities at $z = \infty \in \Vinfeps$. In order for such transformations to arise from a $T$ that is independent of $t$ with $d = - T + \rho_{\infty}$, we therefore require that the remaining coefficients $A$, $B$, $C$, $E$, $G$ and $J$ must be independent of $t$ (i.e. constant). An analysis of the explicit form of these coefficients given above implies shows that this condition is only possible if
\[
\frac{\dot{\lambda}}{\lambda^3} = \frac{k}{2}, \qquad \dot{\alpha} = \dot{\beta} = 0,
\]
where $k$ is a constant. Integrating these equations yields~\eqref{lambdatransformation}. Therefore the only transformation on the one-instanton moduli space that arises from a symmetry of the form~\eqref{Jsymmetry} with $d(t: x, z) = - T(x, z) + \rho_{\infty}(t, x, z)$ is a scaling of the moduli space.
\end{proof}

\begin{remark}
The group of transformations on the one-instanton moduli space is therefore only one-dimensional. Such a collapse to a finite-dimensional action is familiar from the theory of harmonic maps (see, e.g.,~\cite{AJS,GO,JS,UhlenbeckJDG}), where the orbits of the group action are also, generically, of high codimension.
\end{remark}

\section{Final remarks}

Our first main result is Theorem~\ref{TMk}, which states that the tangent space to the instanton moduli spaces, $\mathcal{M}_k$, are generated by symmetries of the {\sdyme}. Nevertheless, our second main result, based on an analysis of the one-instanton moduli space, is that the subgroup of the symmetry group that preserves the $L^2$ nature of the connection, and hence has orbits that lie in a particular $\mathcal{M}_k$, is rather small. In particular, the orbits of this subgroup on the space $\mathcal{M}_k$ are of high codimension. We have restricted ourselves to one-parameter families of ADHM data that arise from transformations of the form~\eqref{Gdot} and~\eqref{Gdotline} with $d(t: x, z) = - T(x, z) + \rho_{\infty}(t, x, z)$. Note that this is a sufficient, but not necessary, condition for equations~\eqref{Gdot} and~\eqref{Gdotline} to be consistent. It is conceivable that there might be a larger group of transformations acting on the moduli spaces, $\mathcal{M}_k$, consistent with these equations, but we have not investigated this possibility.

It is hoped that there is a more elegant way of carrying out the calculations in the previous section. In particular (also regarding the remark in the previous paragraph) one would like to pull the infinitesimal action on the patching matrix~\eqref{Gdot} directly up to the space of ADHM data. An alternative approach to extending our analysis would be to investigate our approach from the point of view of Donaldson's reformulation of the ADHM construction~\cite{Don}, where one views instantons as defining holomorphic bundles over {\cptwo}. Restricting our constructions to the {\cptwo} picture is straightforward, but it is again to directly calculate the action of the symmetry transformations on the data. Work of Nakamura~\cite{N} concerning dynamical systems defined on the space of data of the Donaldson construction may be relevant in this regard. The approach where one might expect the symmetries to have the simplest form would be within Atiyah's reformulation~\cite{At2} of the instanton moduli spaces in terms of holomorphic maps $\cpone \rightarrow \Omega G$. In this case, the connection with harmonic map theory is quite strong. In the case of the {\sdyme}, however, one expects the symmetry group to act directly on the map in the Atiyah construction, whereas for harmonic maps the \lq\lq dressing action\rq\rq\ acts purely on the space $\Omega G$. It is also quite difficult to see directly how the action on the patching matrix or ADHM data transfer to the Atiyah picture, due to the non-holomorphic transformations required in passing from the ADHM construction to this approach.

More broadly, thinking of $(\lambda, \alpha, \beta)$ as coordinates on the five-dimensional ball (with $(\alpha, \beta)$ compactified to the four-sphere and $\lambda$ the radial coordinate) then the flow in~\eqref{lambdatransformation} is simply a radial scaling. In particular, for $k > 0$, the flow converges to the fixed point $\lambda = 0$ as $t \rightarrow -\infty$, and diverges to $+\infty$ as $t \rightarrow \left( \frac{1}{k \lambda^2} \right)-$. Such flows are, in some respects, reminiscent of Morse flows, and it would be of interest to know whether our approach has a Morse-theoretic interpretation. In addition, it would be interesting to relate our work to other examples of systems where one has a symmetry algebra, but no corresponding group action e.g. Teichm\"{u}ller theory\footnote{The author is grateful to Prof. K.\ Ono for this suggestion.}.

\vs
As mentioned in the Introduction, the original motivation for this work was to determine whether the integrable systems approach to the {\sdyme} could give information about instanton moduli spaces as used in the more topological context of Donaldson theory. In this regard, the results of this paper should be viewed alongside the results of the companion paper~\cite{reducible}. In~\cite{reducible}, reducible connections were studied on open subsets of $\mathbb{R}^4$, and were found to bear a strong resemblance to harmonic maps of finite type (see, e.g., \cite[Chapter~24]{Gu}). In particular, all reducible connections lie in the orbit, under flows~\eqref{Gdot}, of the flat connection on $\mathbb{R}^4$. Therefore instanton solutions on $\mathbb{R}^4$ and reducible connections (which are necessarily not $L^2$ on $\mathbb{R}^4$) appear to have quite different behaviour under the symmetry group of the {\sdyme}. Since reducible and irreducible connections play a different role in Donaldson's work~\cite{Don1}, corresponding to the smooth and singular parts of the moduli space respectively, it is striking that such connections also seem to have different behaviour from the point of view of integrable systems. In this respect, it would be of particular interest to investigate the one-instanton moduli space on {\cptwo}, where one has $L^2$ and reducible connections in the same moduli space.

\appendix
\section{Action of symmetries on the patching matrix}
\label{App:symm}

It appears that the direct derivation of the infinitesimal flow, \eqref{Gdot}, from the flow of the $J$-function, \eqref{Jsymmetry}, has not appeared in the literature. We therefore give a proof of this result here. The closest to our derivation that we have found is the corresponding construction for harmonic maps into Lie groups given in~\cite[\S3-4]{UhlenbeckJDG}.

\vs
For ease of notation, we define the quantities
\begin{subequations}
\begin{align}
\alpha(x, \lambda) &:= \Psi_{\infty}(x, \lambda) T(x, \lambda) \Psi_{\infty}(x, \lambda)^{-1},
\\
\alpha(x, \lambda)^{\dagger} &:= \Psi_0(x, \sigma(\lambda)) T(x, \lambda)^{\dagger} \Psi_0(x, \sigma(\lambda))^{-1},
\end{align}\label{alpha}\end{subequations}
and recall the solution of the linearisation equation, \eqref{Jsymmetry}, in this notation:
\begin{equation}
\dot{J}(x) = \psi_{\infty}(x)^{-1}
\left[ \alpha(x, \lambda) + \alpha(x, \lambda)^{\dagger} \right] \psi_0(x).
\label{Jsymmetry2}
\end{equation}

\begin{proposition}
There exists a function $h_{\infty}(x, z) \equiv h_{\infty}(u - z \vbar, v + z \ubar, z)$ with the property that
\begin{align}
&\dot{\Psi}_{\infty}(x, z) \Psi_{\infty}(x, z)^{-1} - \dot{\psi}_{\infty}(x) \psi_{\infty}(x)^{-1} =
\frac{\lambda}{\lambda - z} \left( \alpha(x, \lambda) - \alpha(x, z) \right)
\nonumber\\
&\hskip2cm+ \frac{1}{1 + z \overline{\lambda}} \left( \alpha(x, \lambda)^{\dagger} - \alpha(x, \sigma(z))^{\dagger} \right)
- \Psi_{\infty}(x, z) h_{\infty}(z) \Psi_{\infty}(x, z)^{-1},
\label{psiinfdot}
\end{align}
for all $z \in \cpone$ such that $z \neq 0, \lambda, \left. -1 \right/ \overline{\lambda}$. Similarly, there exists a function $h_0(x, z) \equiv h_0(u - z \vbar, v + z \ubar, z)$ such that
\begin{align}
&\dot{\Psi}_0(x, z) \Psi_0(x, z)^{-1} - \dot{\psi}_0(x) \psi_0(x)^{-1} =
\frac{z}{\lambda - z} \left( \alpha(x, \lambda) - \alpha(x, z) \right)
\nonumber\\
&\hskip2.2cm- \frac{z \overline{\lambda}}{1 + z \overline{\lambda}} \left( \alpha(x, \lambda)^{\dagger} - \alpha(x, \sigma(z))^{\dagger} \right)
+ \Psi_0(x, z) h_0(z) \Psi_0(x, z)^{-1},
\label{psi0dot}
\end{align}
for all $z \in \cpone$ such that $z \neq \infty, \lambda, \left. -1 \right/ \overline{\lambda}$.
\end{proposition}

\begin{proof}
{}From~\eqref{Jsymmetry2}, we deduce that
\begin{equation}
\dot{\psi}_0(x) {\psi}_0(x)^{-1} - \dot{\psi}_{\infty}(x) {\psi}_{\infty}(x)^{-1} =
\alpha(x, \lambda) + \alpha(x, \lambda)^{\dagger}.
\label{a4}
\end{equation}
{}From the defining relations for $\psi_0(x, z), \psi_{\infty}(x, z)$ we deduce that the derivative of the components of the connection are given by
\begin{align*}
\left( \dot{A}_{\ubar} - z \dot{A}_v \right) &= - \left( D_{\ubar} - z D_v \right) \left( \dot{\Psi_0}(x, z) \Psi_0(x, z)^{-1} \right)
\\
&\hskip 3cm = - \left( D_{\ubar} - z D_v \right) \left( \dot{\Psi}_{\infty}(x, z) \Psi_{\infty}(x, z)^{-1} \right),
\\
\left( \dot{A}_{\vbar} + z \dot{A}_u \right) &= - \left( D_{\vbar} + z D_u \right) \left( \dot{\Psi_0}(x, z) \Psi_0(x, z)^{-1} \right)
\\
&\hskip 3cm = - \left( D_{\vbar} + z D_u \right) \left( \dot{\Psi}_{\infty}(x, z) \Psi_{\infty}(x, z)^{-1} \right).
\end{align*}
This expression implies that
\begin{align*}
\left( D_{\ubar} - z D_v \right) \left[ \dot{\Psi}_{\infty}(x, z) \Psi_{\infty}(x, z)^{-1} \right]
& = D_{\ubar} \left[ \dot{\psi}_0(x) \psi_0(x)^{-1} \right]
\\
&\hskip 2.5cm - z D_v \left[ \dot{\psi}_{\infty}(x) \psi_{\infty}(x)^{-1} \right],
\\
\left( D_{\vbar} + z D_u \right) \left[ \dot{\Psi}_{\infty}(x, z) \Psi_{\infty}(x, z)^{-1} \right]
& = D_{\vbar} \left[ \dot{\psi}_0(x) \psi_0(x)^{-1} \right]
\\
&\hskip 2.5cm + z D_u \left[ \dot{\psi}_{\infty}(x) \psi_{\infty}(x)^{-1} \right].
\end{align*}
We need to solve these equations for $\Psi_{\infty}(x, z)$ with the boundary condition that $\dot{\Psi}_{\infty}(x, z) \rightarrow \dot{\psi}_{\infty}(x)$ as $z \rightarrow \infty$.
These equations may be rewritten in the form
\begin{align*}
\left( D_{\ubar} - z D_v \right)\!\left[ \dot{\Psi}_{\infty}(x, z) \Psi_{\infty}(x, z)^{-1} - \dot{\psi}_{\infty}(x) \psi_{\infty}(x)^{-1} \right]
& = D_{\ubar}\!\left[ \alpha(x, \lambda) + \alpha(x, \lambda)^{\dagger} \right]\!,
\\
\left( D_{\vbar} + z D_u \right)\!\left[ \dot{\Psi}_{\infty}(x, z) \Psi_{\infty}(x, z)^{-1} - \dot{\psi}_{\infty}(x) \psi_{\infty}(x)^{-1} \right]
& = D_{\vbar}\!\left[ \alpha(x, \lambda) + \alpha(x, \lambda)^{\dagger} \right]\!.
\end{align*}
We now note that
\[
\left( D_{\ubar} - \lambda D_v \right) \alpha(x, \lambda) = \left( D_{\vbar} + \lambda D_u \right) \alpha(x, \lambda) = 0.
\]
Therefore, for all $z \neq \lambda$,
\begin{align*}
D_{\ubar} \alpha(x, \lambda) &=
\frac{\lambda}{\lambda - z} \left( D_{\ubar} - z D_v \right) \alpha(x, \lambda),
\\
D_{\vbar} \alpha(x, \lambda) &=
\frac{\lambda}{\lambda - z} \left( D_{\vbar} + z D_u \right) \alpha(x, \lambda).
\end{align*}
Similarly,
\[
\left( D_v + \overline{\lambda} D_{\ubar} \right) \alpha(x, \lambda)^{\dagger} =
\left( D_u - \overline{\lambda} D_{\vbar} \right) \alpha(x, \lambda)^{\dagger} = 0,
\]
from which we deduce that, for all $z \neq \left. - 1 \right/ \overline{\lambda}$,
\begin{align*}
D_{\ubar} \alpha(x, \lambda)^{\dagger} &=
\frac{1}{1 + z \overline{\lambda}} \left( D_{\ubar} - z D_v \right) \alpha(x, \lambda)^{\dagger},
\\
D_{\vbar} \alpha(x, \lambda)^{\dagger} &=
\frac{1}{1 + z \overline{\lambda}} \left( D_{\vbar} + z D_u \right) \alpha(x, \lambda)^{\dagger}.
\end{align*}
Hence,
\begin{align*}
&\left( D_{\ubar} - z D_v \right)
\left[ \dot{\Psi}_{\infty}(x, z) \Psi_{\infty}(x, z)^{-1} - \dot{\psi}_{\infty}(x) \psi_{\infty}(x)^{-1} \right.
\\
&\hskip 6cm \left. - \frac{\lambda}{\lambda - z} \alpha(x, \lambda) - \frac{1}{1 + z \overline{\lambda}} \alpha(x, \lambda)^{\dagger} \right] = 0,
\end{align*}
and, similarly, $\left( D_{\vbar} + z D_u \right) \left[ \dots \right] = 0$. It then follows that there exists a function $H_{\infty}(u-z\vbar, v+z\ubar, z)$ with the property that
\begin{align}
\dot{\Psi}_{\infty}(x, z) \Psi_{\infty}(x, z)^{-1} - \dot{\psi}_{\infty}(x) &\psi_{\infty}(x)^{-1} =
\frac{\lambda}{\lambda - z} \alpha(x, \lambda)
\nonumber\\
&+ \frac{1}{1 + z \overline{\lambda}} \alpha(x, \lambda)^{\dagger}  - \Psi_{\infty}(x, z) H_{\infty}(z) \Psi_{\infty}(x, z)^{-1},
\label{poled}
\end{align}
Taking
\[
H_{\infty}(x, z) = h_{\infty}(x, z) - \Psi_{\infty}(x, z)^{-1} \left[ \frac{\lambda}{\lambda - z} \alpha(x, z) + \frac{1}{1 + z \overline{\lambda}} \alpha(x, \sigma(z))^{\dagger} \right] \Psi_{\infty}(x, z)
\]
cancels the poles in the first two terms in the right-hand-side of~\eqref{poled}, and yields Equation~\eqref{psiinfdot}. A similar argument for $\Psi_0(x, z)$ yields equation~\eqref{psi0dot}.
\end{proof}

\begin{lemma}
\[
\dot{G}(z) = T(z) G(z) + G(z)  T^*(z) + h_{\infty}(z) G(x, z) + G(x, z) h_0(z).
\]
\label{lem:Gdot}
\end{lemma}
\begin{proof}
Firstly,
\begin{align*}
\dot{G}(z) &= \frac{\partial}{\partial s} \left[ \Psi_{\infty}(x, z)^{-1} \cdot \Psi_0(x, z) \right]
\\
&= \Psi_{\infty}(x, z)^{-1} \left[ \dot{\Psi}_0(x, z) \cdot \Psi_0(x, z)^{-1} - \dot{\Psi}_{\infty}(x, z) \cdot \Psi_{\infty}(x, z)^{-1} \right] \Psi_0(x, z).
\end{align*}
Now use equations~\eqref{alpha}, \eqref{psiinfdot}, \eqref{psi0dot} and~\eqref{a4}.
\end{proof}

The left-hand-side of~\eqref{psi0dot} is analytic for $|z| < 1 + \epsilon$. Any singularities in this region that occur in the first two terms on the right-hand-side must therefore be cancelled by corresponding singularities in the function $h_0$. It turns out that this consideration is enough to determine $h_0$ up to addition of a function of $(u-z\vbar, v+z\ubar, z)$ that is holomorphic on the region $|z| < 1 + \epsilon$. Similar remarks apply to $h_{\infty}$ and equation~\eqref{psiinfdot}.

\begin{proposition}
There exists a function ${\rho}_0(x, z) \equiv {\rho}_0(u - z \vbar, v + z \ubar, z)$, holomorphic for $|z| < 1 + \epsilon$ with the property that on the region $\frac{1}{1 + \epsilon} < |z| < 1 + \epsilon$ we have
\begin{align}
&\dot{\Psi}_0(x, z) \Psi_0(x, z)^{-1} - \dot{\psi}_0(x) \psi_0(x)^{-1} =
\frac{1}{\lambda - z} \left( z \alpha(x, \lambda) - \lambda \alpha(x, z) \right)
\nonumber\\
&\hskip 2cm - \frac{1}{1 + z \overline{\lambda}} \left( z \overline{\lambda} \alpha(x, \lambda)^{\dagger} + \alpha(x, \sigma(z))^{\dagger} \right)
+ \Psi_0(x, z) {\rho}_0(z) \Psi_0(x, z)^{-1}.
\label{psi0dot2}
\end{align}
\end{proposition}
\begin{proof}
Rearranging equation~\eqref{psi0dot} yields
\begin{align*}
h_0(z) &= \chi_0(z)^{-1} \dot{\chi}_0(z) - \frac{z}{\lambda - z} \Psi_0(z)^{-1} \left( \alpha(x, \lambda) - \alpha(x, z) \right) \Psi_0(z)
\\
&\hskip 2cm + \frac{z \overline{\lambda}}{1 + z \overline{\lambda}} \Psi_0(z)^{-1} \left( \alpha(x, \lambda)^{\dagger} - \alpha(x, \sigma(z))^{\dagger} \right) \Psi_0(z).
\end{align*}
Since $\Psi_0$ is analytic for $|z| < 1 + \epsilon$ and the poles at $z = \lambda, \sigma(\lambda)$ have been cancelled, it follows that $h_0$ is analytic for $\frac{1}{1 + \epsilon} < |z| < 1 + \epsilon$. We may therefore split $h_0(z) = h_0^{(0)}(z) + h_0^{(\infty)}(z)$ where $h_0^{(0)}$ is analytic for $|z| < 1 + \epsilon$ and $h_0^{(\infty)}$ is analytic for $|z| > \frac{1}{1 + \epsilon}$. For $|z| > \frac{1}{1 + \epsilon}$, we have
\[
h_0^{(\infty)}(z) = - \frac{1}{2 \pi i} \oint_{\gamma_-} \frac{h_0(w)}{w-z} dw,
\]
where $\gamma_- = \{ w \in \mathbb{C}: w = \frac{1}{1 + \epsilon'} \}$, where $\epsilon' < \epsilon$ is chosen such that $|z| > \frac{1}{1 + \epsilon'}$. Using the fact that $\chi$ and $\Psi_0$ are analytic for $|z| < \frac{1}{1 + \epsilon'}$, we find that
\[
h_0^{(\infty)}(z) = \frac{1}{2 \pi i} \oint_{\gamma_-} \frac{1}{w-z} \left[ \frac{w \overline{\lambda}}{1 + w \overline{\lambda}} T^*(w) - \frac{w}{\lambda - w} G(w)^{-1} T(w) G(w) \right] dw
\]
for $|z| > \frac{1}{1 + \epsilon'}$. Differentiating under the integral sign, we find that
\[
\left( \partial_{\ubar} - z \partial_v \right) h_0^{(\infty)}(z) = \partial_{\ubar} K(x), \qquad
\left( \partial_{\vbar} + z \partial_u \right) h_0^{(\infty)}(z) = \partial_{\vbar} K(x)
\]
where
\[
K(x) := \frac{1}{2 \pi i} \oint_{\gamma_-} \left[ \frac{1}{w - \sigma(\lambda)} T^*(w) + \frac{1}{w - \lambda} G(w)^{-1} T(w) G(w) \right] dw.
\]
Note that this expression is independent of $z$. In order to construct the function $h_0$, we must find a function $h_0^{(0)}$, holomorphic (in $z$) for $|z| < 1 + \epsilon'$ with the property that
\[
\left( \partial_{\ubar} - z \partial_v \right) h_0^{0}(z) = - \partial_{\ubar} K(x), \qquad
\left( \partial_{\vbar} + z \partial_u \right) h_0^{0}(z) = - \partial_{\vbar} K(x).
\]
To construct such a function, we define the contour $\gamma_+ = \{ w \in \mathbb{C}: |w| = 1 + \epsilon' \}$ and deduce that
\begin{align*}
K(x) &= \frac{1}{2 \pi i} \oint_{\gamma_+} \left[ \frac{1}{w - \sigma(\lambda)} T^*(w) + \frac{1}{w - \lambda} G(w)^{-1} T(w) G(w) \right] dw
\\
&\hskip 6cm- T^*(\sigma(\lambda)) - G(\lambda)^{-1} T(\lambda) G(\lambda).
\end{align*}
We then find that, for $|z| < 1 + \epsilon'$
\begin{align*}
- \partial_{\ubar} K(x) &= - \frac{1}{2 \pi i} \oint_{\gamma_+} \left[ \frac{1}{w - \sigma(\lambda)} \partial_{\ubar} T^*(w) + \frac{1}{w - \lambda} \partial_{\ubar} \left( G(w)^{-1} T(w) G(w) \right) \right] dw
\\
& \hskip 2cm + \partial_{\ubar} T^*(\sigma(\lambda)) + \partial_{\ubar} \left( G(\lambda)^{-1} T(\lambda) G(\lambda) \right)
\\
&= \left( \partial_{\ubar} - z \partial_v \right) \Phi(x, \lambda, z),
\end{align*}
where
\begin{align*}
\Phi(x, \lambda, z) &:= -  \frac{1}{2 \pi i} \oint_{\gamma_+} \!\frac{w}{w-z} \!\left[ \frac{1}{w - \sigma(\lambda)} T^*(w) + \frac{1}{w - \lambda} \!\left( G(w)^{-1} T(w) G(w) \right)\! \right] \!dw
\\
& \hskip 3cm + \frac{\sigma(\lambda)}{\sigma(\lambda)-z} T^*(\sigma(\lambda)) + \frac{\lambda}{\lambda - z} G(\lambda)^{-1} T(\lambda) G(\lambda),
\end{align*}
with a similar expression for $- \partial_{\ubar} K(x)$. Again cancelling the poles at $z = \lambda, \sigma(\lambda)$, we deduce that, for $|z| < 1 + \epsilon'$, we may take
\begin{align*}
h_0^{(0)}(z) =&  \rho_0(z) + \frac{\lambda}{\lambda - z} \left( G(\lambda)^{-1} T(\lambda) G(\lambda) - G(z)^{-1} T(z) G(z) \right)
\\
&- \frac{1}{2 \pi i} \oint_{\gamma_+} \frac{w}{w-z} \left[ \frac{1}{w - \sigma(\lambda)} T^*(w) + \frac{1}{w - \lambda} \left( G(w)^{-1} T(w) G(w) \right) \right] dw
\\
&+ \frac{\sigma(\lambda)}{\sigma(\lambda)-z} \left[ T^*(\sigma(\lambda)) - T^*(z) \right]  ,
\end{align*}
where $\rho_0 = \rho_0(u-z\vbar, v+\ubar, z)$ is analytic for $|z| < 1 + \epsilon'$. Finally, we note that, in the region $\frac{1}{1 + \epsilon} < |z| < 1 + \epsilon$ we have
\begin{equation}
h_0(z) = h_0^{(0)}(z) + h_0^{(\infty)}(z) = \frac{z}{\lambda - z} G(z)^{-1} T(z) G(z) - \frac{z \overline{\lambda}}{1 + z \overline{\lambda}} T^*(z) + \rho_0(z).
\label{h0}
\end{equation}
Substituting this expression into~\eqref{psi0dot} yields~\eqref{psi0dot2}.
\end{proof}

\begin{theorem}
On the region $\frac{1}{1 + \epsilon} < |z| < 1 + \epsilon$ we have
\[
\dot{G}(z) = - T(z) G(z) - G(z) T^*(z) + \rho_{\infty}(z) G(x, z) + G(x, z) \rho_0(z).
\]
\end{theorem}
\begin{proof}
The reality conditions for $\Psi_0$ and $\Psi_{\infty}$ imply that $h_{\infty}(z) = h_0^*(z)$. The result then follows from Lemma~\ref{lem:Gdot} and equation~\eqref{h0}.
\end{proof}

\begin{remark}
Since the functions $\rho_0, \rho_{\infty}$ are holomorphic in $(u-z\vbar, v+z\ubar, z)$ and analytic for $|z| < 1 + \epsilon$, $|z| > \frac{1}{1 + \epsilon}$, respectively, they simply generate holomorphic changes of basis on these regions. As such, modulo holomorphic changes of basis, the symmetry~\eqref{Jsymmetry} generates the flow
\[
\dot{G}(z) = - T(z) G(z) - G(z) T^*(z)
\]
for the patching matrix. Since $T$ is independent of $t$, the corresponding one-parameter group of transformations determined by $T$ with initial conditions the patching matrix $G_0(x, z)$ is of the form
\[
G(t; x, z) = \exp \left( - t T(x, z) \right) G_0(x, z) \exp \left( - t T^*(x, z) \right).
\]
\end{remark}

In particular, we recover the group action constructed on heuristic grounds by Crane~\cite{Cr}: Given a map $h\colon X \times S^1 \rightarrow \SLtwoC$ that extends to a holomorphic map $\tilde{h}\colon X \times \Veps \rightarrow  \SLtwoC$ (where holomorphic means with respect to the complex structure $X \times \Veps$ as a subset of $\cpthree$) then the group action on patching matrix is of the form
\[
G(x, z) \mapsto \left( h \cdot G \right)(x, z) := \tilde{h}(x, z) G(x, z) \tilde{h}^*(x, z).
\]

\def\cprime{$'$}

\end{document}